\documentclass[11pt]{article}
\usepackage[margin=1in]{geometry}
\usepackage{amsfonts}
\usepackage{latexsym,amssymb,amsfonts,graphicx}
\usepackage{amsmath,dsfont}
\usepackage[linewidth=1pt]{mdframed}
\usepackage{mathrsfs}
\usepackage[normalem]{ulem}
\usepackage{epsfig}
\usepackage{tabularx}
\usepackage{geometry}
\usepackage{verbatim}
\usepackage{enumitem}
\usepackage{framed}
\usepackage{algorithm2e}

\usepackage{url}
\usepackage{bm}
\usepackage{color}
\usepackage{natbib}

\usepackage[unicode=true,colorlinks,citecolor=linkcolor,linkcolor=blue,urlcolor=blue]{hyperref}

\providecommand{\U}[1]{\protect\rule{.1in}{.1in}}

\newtheorem{thm}{Theorem}[section]

\newtheorem{Assumption}{\bf Assumption}

\newtheorem{lem}{Lemma}[section]

\newtheorem{prop}{Proposition}[section]

\usepackage[bf,SL,BF]{subfigure}

\newenvironment{proof}[1][Proof]{\noindent\textbf{#1.} }{\ \rule{0.5em}{0.5em}}
\numberwithin{equation}{section}
\allowdisplaybreaks[4]

\definecolor{linkcolor}{rgb}{0,0,0.7}

\begin{document}

\title{Optimal portfolio under ratio-type periodic  evaluation in incomplete markets with stochastic factors}
\author{Wenyuan Wang\thanks{School of Mathematics and Statistics, Fujian Normal University, Fuzhou, 350007, China; and School of Mathematical Sciences, Xiamen University, Xiamen, 361005, China. Email: wwywang@xmu.edu.cn}
\and
Kaixin Yan\thanks{Corresponding author. School of Mathematical Sciences, Xiamen University, Xiamen, 361005, China. Email: kaixinyan@stu.xmu.edu.cn}
\and
Xiang Yu\thanks{Department of Applied Mathematics, The Hong Kong Polytechnic University, Kowloon, Hong Kong. E-mail: xiang.yu@polyu.edu.hk}
}

\date{\ }

\maketitle
\vspace{-0.6in}

\begin{abstract}
This paper studies a type of periodic utility maximization for portfolio management in an incomplete market model, where the underlying price diffusion process depends on some external stochastic factors. The portfolio performance is periodically evaluated on the relative ratio of two adjacent wealth levels over an infinite horizon. For both power and logarithmic utilities, we formulate the auxiliary one-period optimization problems with modified utility functions, for which we develop the martingale duality approach to establish the existence of the optimal portfolio processes and the dual minimizers can be identified as the ``least favorable" completion of the market. With the help of the duality results in the auxiliary problems and some fixed point arguments, we further derive and verify the optimal portfolio processes in a periodic manner for the original periodic evaluation problems over an infinite horizon.\\
\ \\
\textbf{Keywords}: Periodic evaluation, relative portfolio performance, incomplete market, stochastic factors, convex duality, utility maximization\\
\
\\
\noindent\textbf{AMS subject classifications.} 91B28,~93E20,~60G44
\end{abstract}

\section{Introduction}
It has been well documented that for the portfolio management by institutional managers, the long term portfolio performance is often exercised to dictate the daily decision making for the fund management. Various long-run criteria for portfolio optimization have been proposed and studied in the literature, leading to several well-known stochastic control and optimizaiton problems over a large or an infinite horizon. The long term optimal growth rate, also named Kelly's criterion, has been popularized thanks to its tractability and simple financial implications. As an important extension, the so-called risk-sensitive portfolio management was introduced by \cite{BPK99} and \cite{FlemingS99} to encode diverse risk attitudes. Later, \cite{Pham03} formulated the long-run outperformance criterion as a large deviation probability control problem to incorporate the benchmark tracking for the fund management. On the other hand, to overcome the limitation of the prescribed time horizon from the conventional paradigm of utility maximization, \cite{MZ08}, \cite{MZ09}, \cite{MZ10} proposed the forward performance measure such that the decision maker's risk preference can be consistently extended forward in time, leading to the optimal portfolio management with an arbitrary trading horizon.

Another new long-run portfolio criterion was proposed in the recent study \cite{TZ23}, which suggests to roll over the same utility sequentially for infinite periods. In  \cite{TZ23}, the utility during each period is generated by the difference between the wealth levels at two adjacent evaluation dates. In particular, the S-shaped utility with the same risk aversion parameter was adopted in \cite{TZ23} to accommodate both cases that the current wealth may outperform or underperform the benchmark level from the preceding evaluation date. This type of periodic evaluation can partially match with the practical exercises in the annual appraisal review in the fund industry. However, as a common consequence of the S-shaped utility, when the current wealth level falls below the benchmark of the previous evaluation date during some time period, the optimal solution may suggest to cease the investment in the risky asset. Recently, \cite{WYY23} extends the periodic evaluation formulation in \cite{TZ23} by considering the periodic relative performance on the ratio between the current wealth level and the one at the previous evaluation date. One advantage of the ratio-type relative performance is that the bankruptcy opportunity is priorly ruled out from the admissible portfolio processes, yielding the optimal portfolio strategy truly supported over the infinite horizon. 

The present paper aims to extend the study in \cite{WYY23} from the Black-Scholes model to an incomplete market model, where the price process of the underlying risky asset follows a diffusion process affected by some correlated stochastic factors. The stochastic factor model has been extensively used in portfolio choice to integrate the asset predictability of the return and the stochastic volatility. For example, some early studies on the predictability of stock returns using the stochastic factor can be found in \cite{Fama77}, \cite{Ferson93}, \cite{Bre97}, \cite{Bre98}, \cite{Camp99},  \cite{Watchter} among others. The optimal investment under the stochastic volatility  or combining the stochastic returns and stochastic volatility, to name a few, was investigated in \cite{French87}, \cite{Kim96},  \cite{Scruggs}, \cite{FPS2000}, \cite{ZT01}, \cite{Pham02}, \cite{FH2003}, \cite{Chacko05}, \cite{Kraft05}, \cite{CH05}, \cite{Liu}, \cite{F13}, \cite{Hata18}, \cite{F20}, and many subsequent studies in different contexts.

Inspired by \cite{TZ23}, we first reformulate the infinite horizon optimization problem into an auxiliary one-period portfolio optimization problem based on dynamic programming principle, leading to two subsequent open questions: (1). The existence and the characterization of the optimal portfolio for the auxiliary problem; (2). The fixed point argument to characterize the original value function and the verification proof of the constructed optimal portfolio for the original problem over the infinte horizon. Comparing with  \cite{TZ23} and \cite{WYY23} in the Black-Scholes model, the unhedgeable risk from the external stochastic factor renders these two open questions significantly more challenging. Take the power utility for example. Firstly, to cope with the auxiliary portfolio optimization problem under a modified utility function (see the induced utility in \eqref{2.8} and \eqref{problem3}), we need to develop the duality method for the incomplete market model in the face of infinitely many dual processes (see \eqref{dual.pro}). It is also a well-known challenge to characterize the dual optimizer in a more explicit manner when our modified utility does not fit the special types such as the standard power utility. Secondly, due to the presence of the stochastic factor, the fixed point in  \eqref{A*.def} depends on the variable $y$ instead of the constant fixed point in \cite{TZ23} and \cite{WYY23}. More importantly, the verification proof for the constructed portfolio using the result from the auxiliary problem also becomes more technical as the we again encounter some duality arguments in the incomplete market model.  

The present paper contributes to the methodology in resolving the aforementioned challenges. To address the optimal solution of the auxiliary problem in the first question, we study the dual problem in the framework of stochastic factor models similar to \cite{CH05} and develop some novel methods for the modified utility to identify the dual optimizer in the sense of ``least favorable" completion of the market. In particular, some new technical proofs are presented for two distinct cases when risk aversion parameter $0<\alpha<1$ (see Proposition \ref{existence.1}) and $\alpha<0$ (see Proposition \ref{existence.2}), respectively. To tackle the difficulties in the second question, we show the existence of the unique fixed point by the contraction mapping in the space of bounded and continuous nonnegative functions. Based on the identification of the dual minimizer in the auxiliary problem and the fixed point result, we are able to verify the optimality of the concatenated wealth process period by period (see Theorem \ref{thm3.1}) using some duality  convergence arguments. For the case when the periodic evaluation is conducted under a logarithmic utility, we provide some simplified proofs for the existence of the optimal portfolio.

The remainder of this paper is organized as follows. Section \ref{Sec.2} introduces the incomplete market model where the risky asset price is governed by a stochastic factor model. The problem formulation under the ratio-type periodic evaluation on the relative portfolio performance is presented afterwards. In Section \ref{sec:power}, we study the case under power utility and develop the martingale duality arguments for the auxiliary one-period optimization problem and complete the verification proof for the constructed optimal periodic portfolio strategy. Section \ref{sec:log} presents the main results and some distinct proofs for the easier case under logarithmic utility.

\section{Market Model and Problem Formulation}\label{Sec.2}
Let $(\Omega,\mathcal{F},\{\mathcal{F}\}_{t\geq0},\mathbb{P})$ be a filtered probability space supporting a two-dimensional Brownian motion $W=(W_{1t}, W_{2t})_{t\geq0}$. The financial market consists of one risky asset $S=(S_t)_{t\geq0}$ and one risk-free asset $B=(B_t)_{t\geq0}$. The continuous compounding interest rate of the risk-free asset is $r\geq 0$ and $B_t=e^{rt}$.
The price process of the risky asset $S=(S_t)_{t\geq0}$ satisfies the stochastic differential equation
\begin{eqnarray}\label{SDE2.1}
\frac{dS_t}{S_t}=\mu(Y_t) dt+\sigma(Y_t) dW_{1t},\quad t\geq 0,
\end{eqnarray}
with $S_0=1$, where $\mu(\cdot)$ and $\sigma(\cdot)>0$ are the drift and volatility respectively. 
Here, we consider the incomplete market model where the drift and volatility depend on an external stochastic factor process $Y_t$, which is governed by
\begin{eqnarray}\label{SDE2.2}
dY_t=b(Y_t)dt+\beta(Y_t)(\rho dW_{1t}+\sqrt{1-\rho^2}dW_{2t}),\quad t\geq 0,
\end{eqnarray}
where $Y_0=y\in\mathbb{R}$, $|\rho|\leq1$, $\beta(\cdot)\neq0$, and, $b(\cdot)$ is the drift of the process $Y_t$. 
To guarantee that a unique solution to \eqref{SDE2.1} and \eqref{SDE2.2} exists respectively, we assume that 
$\mu(\cdot)$, $\sigma(\cdot)$, $b(\cdot)$ and $\beta(\cdot)$ satisfy the global Lipschitz and linear growth conditions
\begin{eqnarray}
|f(x)-f(y)|\leq K|x-y|,\quad x,y\in\mathbb{R},\nonumber\\
f^2(y)\leq K^2(1+y^2),\quad y\in\mathbb{R},\nonumber
\end{eqnarray}
where $K$ is a positive constant and $f$ represents $\mu,\sigma,b$ and $\beta$.
In addition, we assume that the drift and volatility coefficients satisfy
\begin{eqnarray}\label{M.def}
\frac{(\mu(Y_t)-r)^2}{\sigma^2(Y_t)}\leq M_0,\quad t\geq0,
\end{eqnarray}
for some positive constant $M_0$.

A trading strategy $\phi=(\phi^S_t,\phi^B_t)_{t\geq0}$ is a predictable process representing the holding of $\phi^S_t$ units of the risky asset and $\phi^B_t$ units of the money market instrument at time $t$. The resulting wealth process $X=(X_t)_{t\geq0}$ is assumed to be self-financing that satisfies 
\begin{eqnarray}
X_t:=\phi^S_t S_t+\phi^B_t B_t=X_0+\int_0^t\phi^S_udS_u+\int_0^t\phi^B_udB_u,\quad t\geq 0,
\end{eqnarray}
with $X_0=x\in\mathbb{R}_+$ being the initial capital. 
We can also write $\pi_t:=\frac{\phi^S_t S_t}{X_t}$ as the proportion of wealth invested in the risky asset, and the dynamics of the resulting wealth process $X=X^{x,y,\pi}$ can be rewritten as
\begin{eqnarray}
dX_t=\pi_tX_t\frac{dS_t}{S_t}+(1-\pi_t)X_t\frac{dB_t}{B_t}=[r+(\mu(Y_t)-r)\pi_t]X_tdt+\sigma(Y_t)\pi_tX_tdW_{1t}, \quad t\geq 0.
\end{eqnarray}

In this paper, we allow both short-selling of stocks and borrowing at the interest rate of the risk-free asset, i.e., the proportion $\pi$ takes values on $\mathbb{R}$. Similar to \cite{TZ23} and \cite{WYY23}, we are interested in a new performance measure based on the periodic evaluation of the wealth $(X_t)_{t\geq0}$ on a sequence of deterministic dates $(T_i)_{i\geq 0}$ with $T_0:=0$. For simplicity, we assume $T_i=i\tau$ for $i\geq 0$ for some constant $\tau>0$ such that the portfolio is evaluated every $\tau$ unit of time (e.g. monthly or annually). We adopt the relative wealth performance through the periodic evaluation in a ratio type as in \cite{WYY23} that
\begin{eqnarray}
\frac{X_{T_i}}{(X_{T_{i-1}})^{\gamma}},\quad i\geq 1,
\end{eqnarray}
for some relative performance parameter $\gamma\in(0,1]$. 

Mathematically speaking, our objective is to solve the optimal portfolio problem over an infinite horizon by employing a periodic evaluation of the relative performance defined by
\begin{eqnarray}
\sup_{X\in\mathcal{U}_0(x,y)}\mathbb{E}\left[\sum_{i=1}^{\infty}e^{-\delta T_i}U\left(\frac{X_{T_i}}{(X_{T_{i-1}})^{\gamma}},Y_{T_i}\right)\right],
\end{eqnarray}
where $\delta>0$ is the agent’s subjective discount factor. 

In the present paper, we only focus on  the power utility function $U(x,y)=\frac{1}{\alpha}x^{\alpha}h(y)$ with $\alpha\in(-\infty,0)\cup(0,1)$ and the logarithmic utility function $U(x,y)=\log x+h(y)$ with $\mathbb{R}\ni y\mapsto h(y)\in\mathbb{R}_+$ being a continuous function defined on the factor level $y$. For simplicity, similar to \cite{ZT01}, it is assumed that $m\leq h(y) \leq 1$, $y\in\mathbb{R}$, for some constant $m\in(0,1)$.

Moreover, the set of admissible portfolio processes is defined by
\begin{eqnarray}\label{admissible_set}
\mathcal{U}_0(x,y)
\hspace{-0.3cm}&:=&\hspace{-0.3cm}
\left\{X:X_t=x+\int^t_0\phi^S_udS_u+\int_0^t\phi^B_udB_u>0 \text{ for all } t\geq0, \text{ $\phi$ is predictable, locally}\right.
\nonumber\\
\hspace{-0.3cm}&&\hspace{0cm}
\left.\text{square-integral and self-financing, and
}\sum_{i=1}^{\infty}e^{-\delta T_i}\mathbb{E}\left[\left(U\left(\frac{X_{T_i}}{(X_{T_{i-1}})^{\gamma}},Y_{T_i}\right)\right)_-\right]<\infty \right\},
\nonumber
\end{eqnarray}
with $a_-:=\max\{-a,0\}$, $x\in\mathbb{R}_+$, and $Y_0=y\in\mathbb{R}$. The admissible set refers to the collection of all non-negative self-financing portfolios that can be generated by an initial wealth $x$ and an initial stochastic factor level $y$. Note that the last condition is of the integrability type, which is needed to ensure the wellposedness of the our portfolio optimization problem under the periodic evaluation performance.

\section{Periodic Evaluation under Power Utility}\label{sec:power}

In this section, we study the previous unconventional portfolio optimization problem when the utility is of the power type. In this regard, let us consider 
\begin{align}
\label{def.U}
    U(x,y):=\frac{1}{\alpha}x^{\alpha}{h(y)},\quad (x,y)\in\mathbb{R}_+\times\mathbb{R},
\end{align}
with the risk aversion coefficient $1-\alpha$ that $\alpha\in(-\infty,0)\cup(0,1)$. 

The goal of the agent is to maximize the total sum of the expected utilities on the ratio-type relative portfolio performance over an infinite periods, and the value function is defined by 
\begin{eqnarray}\label{problem}
V(x,y)
\hspace{-0.3cm}&:=&\hspace{-0.3cm}
\sup_{X\in\mathcal{U}_0(x,y)}\mathbb{E}\left[\frac{1}{\alpha}\sum_{i=1}^{\infty}e^{-\delta T_i}\left(\frac{X_{T_i}}{\left(X_{T_{i-1}}\right)^{\gamma}}\right)^{\alpha}h({Y_{T_i}})\right],
\end{eqnarray}
where $\mathcal{U}_0(x,y)$ is given by \eqref{admissible_set}.
By the Markov property of the stochastic factor model, one can easily derive the following dynamic programming principle that
\begin{eqnarray}\label{ddp}
V(x,y)=\sup_{X\in\mathcal{U}_0(x,y)}\mathbb{E}\left[\frac{1}{\alpha}e^{-\delta T_1}\left(\frac{X_{T_1}}{x^{\gamma}}\right)^{\alpha}h({Y_{T_1}})+e^{-\delta T_1}V(X_{T_1},Y_{T_1})\right].
\end{eqnarray}

For the wellposedness of the problem, the following standing assumption is imposed throughout this section.

\begin{Assumption}
\label{ass1}
The model parameters satisfy $\delta>\zeta(\alpha(1-\gamma)){\vee0}$, where the function $\zeta$ is defined by $\zeta(x):=rx+x M_0/2(1-x)$,  $x\in(-\infty,1)$, with $M_0\in(0,\infty)$ being the pre-specified constant satisfying \eqref{M.def}.
\end{Assumption}

The following proposition gives the upper and lower bounds for the value function $V$ defined in \eqref{problem}.

\begin{prop}
It holds that
\begin{eqnarray}
\label{2.12}
\frac{{\left(m\mathbf{1}_{\{0<\alpha<1\}}+\mathbf{1}_{\{\alpha<0\}}\right)}e^{(r\alpha-\delta)\tau}}{\alpha(1-e^{-(\delta-r\alpha(1-\gamma))\tau})}x^{\alpha(1-\gamma)}\leq V(x,y)\leq \frac{{\left(\mathbf{1}_{\{0<\alpha<1\}}+m\mathbf{1}_{\{\alpha<0\}}\right)}e^{(\zeta(\alpha)-\delta)\tau}}{\alpha(1-e^{(\zeta(\alpha(1-\gamma))-\delta)\tau})}x^{\alpha(1-\gamma)},
\end{eqnarray}
where $m\in(0,1)$ is the lower bound of the function $h(\cdot)$ that appears in the utility function $U(\cdot,\cdot)$ given by \eqref{def.U}.
\end{prop}
\begin{proof}
We first assume that $\alpha\in(0,1)$. The lower bound of $V(x,y)$ can be derived by noting that $X=(xe^{rt})_{t\geq 0}$ is an admissible portfolio in $\mathcal{U}_0(x,y)$.
It is sufficient to derive the upper bound. For any $X\in\mathcal{U}_0(x,y)$ and $n\geq1$, using the upper bound of $h(\cdot)$, one has
\begin{eqnarray}
\label{3.11.v0}
\mathbb{E}\left[\sum_{i=1}^{n}e^{-\delta T_i}\frac{1}{\alpha}\bigg(\frac{X_{T_i}}{X_{T_{i-1}}^{\gamma}}\bigg)^{\alpha}h(Y_{T_i})\right]
\hspace{-0.3cm}&=&\hspace{-0.3cm}
\sum_{i=1}^{n}e^{-\delta T_{i}}\mathbb{E}\left[\frac{1}{\alpha}\bigg(\frac{X_{T_i}}{X_{T_{i-1}}}\bigg)^{\alpha}X_{T_{i-1}}^{\alpha(1-\gamma)}{h(Y_{T_{i}})}\right]
\nonumber\\
\hspace{-0.3cm}&=&\hspace{-0.3cm}
\sum_{i=1}^{n}e^{-\delta T_{i}}\mathbb{E}\left[\mathbb{E}\left[\frac{1}{\alpha}\bigg(\frac{X_{T_i}}{X_{T_{i-1}}}\bigg)^{\alpha}X_{T_{i-1}}^{\alpha(1-\gamma)}{h(Y_{T_{i}})}\bigg|\mathcal{F}_{T_{i-1}}\right]\right]
\nonumber\\
\hspace{-0.3cm}&\leq&\hspace{-0.3cm}
\sum_{i=1}^{n}e^{-\delta T_{i}}\mathbb{E}\left[\mathbb{E}\left[\frac{1}{\alpha}\bigg(\frac{X_{T_i}}{X_{T_{i-1}}}\bigg)^{\alpha}\bigg|\mathcal{F}_{T_{i-1}}\right]X_{T_{i-1}}^{\alpha(1-\gamma)}\right]
\nonumber\\
\hspace{-0.3cm}&\leq&\hspace{-0.3cm}
\sum_{i=1}^{n}e^{-\delta T_{i}}
\mathbb{E}\left[\left(\sup_{X\in\mathcal{U}_{T_{i-1}}(1,Y_{T_{i-1}})}\mathbb{E}\left[\frac{1}{\alpha}X_{T_i}^{\alpha}|\mathcal{F}_{T_{i-1}}\right]\right)X_{T_{i-1}}^{\alpha(1-\gamma)}\right].
\end{eqnarray}
Note that $\sup_{X\in\mathcal{U}_{T_{i-1}}(1,Y_{T_{i-1}})}\mathbb{E}[X^{\alpha}_{T_i}|\mathcal{F}_{T_{i-1}}]$ is the value function of the investment problem considered in \cite{ZT01} with the special choice that $h(\cdot)\equiv1$.
Using Proposition 3.1 of \cite{ZT01}, we have 
\begin{eqnarray}\label{inv.pro}
\sup_{X\in\mathcal{U}_{T_{i-1}}(1,Y_{T_{i-1}})}\mathbb{E}\left[\frac{1}{\alpha}X_{T_i}^{\alpha}\bigg|\mathcal{F}_{T_{i-1}}\right]\leq\frac{1}{\alpha} e^{\zeta(\alpha) \tau}.
\end{eqnarray}
Similarly, it holds that
\begin{eqnarray}
\label{3.13.v0}
\sup_{X\in\mathcal{U}_0(x,y)}\mathbb{E}\left[\frac{1}{\alpha(1-\gamma)}X_{T_{i-1}}^{\alpha(1-\gamma)}\right]\leq \frac{x^{\alpha(1-\gamma)}}{\alpha(1-\gamma)}e^{\zeta(\alpha(1-\gamma))(i-1)\tau}.
\end{eqnarray}
Combining \eqref{3.11.v0}-\eqref{3.13.v0}, we obtain that
\begin{eqnarray}
\mathbb{E}\left[\frac{1}{\alpha}\sum_{i=1}^{n}e^{-\delta T_i}\bigg(\frac{X_{T_i}}{X_{T_{i-1}}^{\gamma}}\bigg)^{\alpha}h(Y_{T_i})\right]
\hspace{-0.3cm}&\leq&\hspace{-0.3cm}
\sum_{i=1}^{n}e^{-\delta T_{i}}
\frac{1}{\alpha}\mathbb{E}\left[\left(\sup_{X\in\mathcal{U}_{T_{i-1}}(1,Y_{T_{i-1}})}\mathbb{E}\left[X_{T_i}^{\alpha}|\mathcal{F}_{T_{i-1}}\right]\right)X_{T_{i-1}}^{\alpha(1-\gamma)}\right]
\nonumber\\
\hspace{-0.3cm}&\leq&\hspace{-0.3cm}
\sum_{i=1}^{n}e^{-\delta T_{i}}e^{\zeta(\alpha)\tau}\sup_{X\in\mathcal{U}_0(x,y)}\mathbb{E}\left[\frac{1}{\alpha}X_{T_{i-1}}^{\alpha(1-\gamma)}\right]
\nonumber\\
\hspace{-0.3cm}&\leq&\hspace{-0.3cm}
\sum_{i=1}^{n}e^{-\delta T_{i}}\frac{1}{\alpha}e^{\zeta(\alpha)\tau}\times x^{\alpha(1-\gamma)}e^{\zeta(\alpha(1-\gamma))(i-1)\tau}.
\end{eqnarray}
The desired upper bound in \eqref{2.12} follows after sending $n\rightarrow\infty$ on both sides.
The proof for the case $\alpha\in(-\infty,0)$ can be easily modified, and it is hence omitted.  
\end{proof}

In view of the scaling property of the utility function $U(x,y)=x^{\alpha}U(1,y)$ and the fact that the value function $V(x,y)$ is bounded with respect to $y$, we heuristically conjecture that our value function admits the form of $V(x,y)={\frac{1}{\alpha}}A^*(y)x^{\alpha(1-\gamma)}$ for some continuous, bounded, and non-negative function $A^*(\cdot)$. Substituting this form of $V$ into \eqref{ddp} and then dividing both sides by ${\frac{1}{\alpha}x^{\alpha(1-\gamma)}}$, one obtains
\begin{eqnarray}\label{A*.def}
A^*(y)
\hspace{-0.3cm}&=&\hspace{-0.3cm}
\alpha\sup_{X\in\mathcal{U}_0(x,y)}\mathbb{E}\left[e^{-\delta T_1}\frac{1}{\alpha}\left(\frac{X_{T_1}}{x}\right)^{\alpha}h({Y_{\tau}})+e^{-\delta T_1}\frac{1}{\alpha}A^*(Y_{T_1})\left(\frac{X_{T_1}}{x}\right)^{\alpha(1-\gamma)}\right]
\nonumber\\
\hspace{-0.3cm}&=&\hspace{-0.3cm}
\alpha\sup_{X\in\mathcal{U}_0(1,y)}\mathbb{E}\left[e^{-\delta \tau}\frac{1}{\alpha}X_{\tau}^{\alpha}h({Y_{\tau}})+e^{-\delta \tau}\frac{1}{\alpha}A^*(Y_{\tau})X_{\tau}^{\alpha(1-\gamma)}\right].
\end{eqnarray}
Hence, in the case of power utility, the characterization of $V$ now simplifies to the characterization of the unknown, non-negative, continuous, and bounded function $A^*(\cdot)$. More importantly, the unknown function $A^*(\cdot)$ will be proven as the the fixed point of a contraction operator defined on the function space ${C}_b^{+}(\mathbb{R})$ consisting of all continuous, bounded, and non-negative functions on $\mathbb{R}$. To this end, we will first establish the existence of the optimizer to the auxiliary one-period optimization problem \eqref{A*.def} for a fixed $A^*$, and then show the existence of the unique fixed point $A^*$ to the operator.

In the presence of stochastic factor, we follow \cite{CH05} to employ the convex duality approach and formulate the associated dual control problem using the idea of the market completion. To this end, let us first introduce some notations. We consider
$$\theta(y):=\frac{\mu(y)-r}{\sigma(y)},\quad y\in\mathbb{R}.$$
Denote by $\mathcal{M}$ the set of progressively measurable processes $(\eta_t)_{t\geq0}$
such that $\mathbb{E}[\int_0^{\tau}\eta_s^2ds]<\infty$ and the local martingale by
\begin{eqnarray}\label{Z.process}
Z_t^{\eta}:=\exp\left(-\int_0^t[\theta(Y_s)dW_{1s}-\eta_sdW_{2s}]-\frac{1}{2}\int_0^t[\theta^2(Y_s)+\eta^2_s]ds\right),\quad t\in[0,\tau],
\end{eqnarray}
is a true martingale for any $y\in\mathbb{R}$. In this vein, for each $\eta\in\mathcal{M}$, we can defined a new probability measure $\mathbb{P}^{\eta}$ by
$d\mathbb{P}^{\eta}:=Z_{\tau}^{\eta}d\mathbb{P}$.  By the Girsanov's theorem, the two-dimensional process $W^{\eta}=(W^{\eta}_{1t},W^{\eta}_{2t})_{t\in[0,\tau]}$ with
$$W_{1t}^{\eta}:=W_{1t}+\int_0^{t}\theta(Y_s)ds\quad\text{and}\quad W_{2t}^{\eta}:=W_{2t}-\int_0^{t}\eta_sds,$$
is a two-dimensional Brownian motion under the measure $\mathbb{P}^{\eta}$. In addition, the dynamics of the two processes $(Y_t)_{t\geq 0}$ and $(Z_{t}^{\eta})_{t\geq 0}$ can be rewritten as
\begin{align}
    dY_t&=\left(b(Y_t)-\beta(Y_t)(\rho\theta(Y_t)-\sqrt{1-\rho^2} \eta_t)\right)dt + \beta(Y_t)(\rho dW_{1t}^{\eta}+\sqrt{1-\rho^2}dW_{2t}^{\eta}),\quad t\geq 0,\nonumber\\
dZ_t^{\eta}&=Z_t^{\eta}\left(\left(\theta^2(Y_t)+\eta^2_{t}\right)dt-\theta(Y_t)dW_{1t}^{\eta}+\eta_tdW_{2t}^{\eta}\right),\quad t\geq 0.
\label{z.process}
\end{align}
One can also verify that 
\begin{eqnarray}\label{dX/B}
d\left[\frac{X_t}{B_t}\right]=\frac{X_t}{B_t}\pi_t\sigma(Y_t)dW^{\eta}_{1t}, \quad t\geq 0.
\end{eqnarray}
Additionally, an application of It\^o's formula to the product of processes $Z^{\eta}$ and $X/B$ yields that
\begin{eqnarray}\label{supermaringale}
\frac{X_tZ_t^{\eta}}{B_t}=x+\int_0^t\frac{X_sZ_s^{\eta}}{B_s}\left[(\pi_s\sigma(Y_s)-\theta(Y_s))dW_{1s}+\eta_sdW_{2s}\right], \quad t\geq 0.
\end{eqnarray}
If $(X_{t})_{t\geq 0}$ is an admissible portfolio process, then, by \eqref{supermaringale}, the process $XZ^{\eta}/B$ is a non-negative $\mathbb{P}$-local martingale, and hence is a $\mathbb{P}$-supermartingale.

Furthermore, let us consider the modified utility function
$\mathbb{R}_+\times\mathbb{R}\ni (x,y)\mapsto h_A(x,y)\in \mathbb{R}$ by
\begin{eqnarray}\label{2.8}
h_A(x,y)
:=\frac{1}{\alpha}x^{\alpha}h(y)+{\frac{1}{\alpha}}A(y)x^{\alpha(1-\gamma)},\quad (x,y)\in \mathbb{R}_+\times\mathbb{R},
\end{eqnarray}
where $A(\cdot)\in C^+_b(\mathbb{R})$ is viewed as a parameter of the function $h_A$.  As a preparation for the main result, we first derive some preliminary properties of the function $h_A$, which plays a pivotal role in later proofs.
\begin{lem}
\label{lem2.1}
The function $h_A(x,y)$ defined by \eqref{2.8} is strictly increasing and strictly concave in $x\in\mathbb{R}_+$, and it holds that $\frac{\partial}{\partial x}h_A(0+,y)=\infty$ and $\frac{\partial}{\partial x}h_A(\infty,y)=0$ for any $y\in\mathbb{R}$. 
In addition, there exist some constants $\vartheta\in(0,1)$ and $\varrho\in(1,\infty)$ such that 
\begin{eqnarray}\label{h_a'>h_a'}
\vartheta \frac{\partial}{\partial x}h_A(x,y)\geq \frac{\partial}{\partial x}h_A(\varrho x,y),\quad (x,y)\in\mathbb{R}_+\times\mathbb{R}.
\end{eqnarray} 
Furthermore, when $\alpha\in(0,1)$, there exist some constants $\kappa_1\in(0,\infty)$ and $\rho_1\in(0,1)$ such that
\begin{eqnarray}\label{0<h}
0< h_A(x,y)\leq \kappa_1(1+x^{\rho_1}),\quad (x,y)\in\mathbb{R}_+\times\mathbb{R};
\end{eqnarray}
when $\alpha\in(-\infty,0)$, there exist some constants $\kappa_2\in(-\infty,0)$ and $\rho_2\in(-\infty,0)$ such that
\begin{eqnarray}\label{0>h}
0>h_A(x,y)\geq \kappa_2(1+x^{\rho_2}),\quad (x,y)\in\mathbb{R}_+\times\mathbb{R}.
\end{eqnarray}
\end{lem}
\begin{proof}
The results follow from elementary calculus. Recall that the function $h(\cdot)\in[m,1]$ with $m\in(0,1)$ and $A(\cdot)\in {C}_b^{+}(\mathbb{R})$. Differentiating twice the both sides of \eqref{2.8} gives
\begin{eqnarray}
\label{2.28.v0}
\frac{\partial}{\partial x}h_A(x,y)= x^{\alpha-1}{h(y)}+A(y)(1-\gamma)x^{\alpha(1-\gamma)-1},\quad (x,y)\in\mathbb{R}_+\times\mathbb{R},
\end{eqnarray}
and
\begin{eqnarray}
\frac{\partial^2}{\partial x^2}h_A(x,y)=(\alpha-1) x^{\alpha-2}{h(y)}+A(y)(1-\gamma)(\alpha-1-\alpha\gamma)x^{\alpha(1-\gamma)-2},\quad (x,y)\in\mathbb{R}_+\times\mathbb{R}.
\end{eqnarray}
Due to the facts of $\gamma\in(0,1]$ and {$\alpha\in(-\infty,0)\cup(0,1)$}, we have $\frac{\partial}{\partial x}h_A(x,y)>0$ and $\frac{\partial^2}{\partial x^2}h_A(x,y)<0$ for any $(x,y)\in\mathbb{R}_+\times\mathbb{R}$, implying that the function $h_A(x,y)$ is strictly concave and strictly increasing  with respect to $x$ on $\mathbb{R}_+$. Furthermore, by \eqref{2.28.v0}, one can easily get $\frac{\partial}{\partial x}h_A(0+,y)=\infty$ and $\frac{\partial}{\partial x}h_A(\infty,y)=0$ for any $y\in\mathbb{R}$. To prove the second claim, for any constant $\varrho\in(1,\infty)$, one can take $\vartheta=\varrho^{\alpha-1}\vee \varrho^{\alpha(1-\gamma)-1}\in(0,1)$. Then 
\begin{align}
    \vartheta \frac{\partial}{\partial x}h_A(x,y)&=\vartheta x^{\alpha-1}h(y)+A(y)(1-\gamma)\vartheta x^{\alpha(1-\gamma)-1}
    \nonumber\\
    &\geq (\varrho x)^{\alpha-1}h(y)+A(y)(1-\gamma)(\varrho x)^{\alpha(1-\gamma)-1}
    \nonumber\\
    &= \frac{\partial}{\partial x}h_A(\varrho x,y),\quad (x,y)\in\mathbb{R}_+\times\mathbb{R}.\nonumber
\end{align}
Finally, when $\alpha\in(0,1)$, taking $\kappa_1=\frac{2}{\alpha}\max\{1,\sup_{y\in\mathbb{R}}A(y)\}\in(0,
\infty)$ and $\rho_1=\alpha\in(0,1)$ yields \eqref{0<h}; when $\alpha\in(-\infty,0)$, taking $\kappa_2=\frac{2}{\alpha}\max\{1,\sup_{y\in\mathbb{R}}A(y)\}\in(-\infty,0)$ and $\rho_2=\alpha\in(-\infty,0)$ yields \eqref{0>h}.
The proof is complete.
\end{proof}

Using similar methods as those used in the proof of Lemma 2.2 of \cite{CH05}, we can now state the original problem \eqref{A*.def} equivalently as
\begin{eqnarray}\label{problem2}
A^*(y)=\Psi(A^*;y),\quad y\in\mathbb{R},
\end{eqnarray}
where the functional ${C}_b^{+}(\mathbb{R})\ni A(\cdot)\mapsto \Psi(A;\cdot)\in {C}_b^{+}(\mathbb{R})$ is defined as
\begin{eqnarray}\label{problem3}
\Psi(A):=\Psi(A;y)\hspace{-0.3cm}&:=&\hspace{-0.3cm}
\alpha e^{-\delta\tau}\sup_{X\in\mathcal{U}_0(1,y)}\mathbb{E}[h_A(X_{\tau},Y_{\tau})]
\nonumber\\
\hspace{-0.3cm}&=&\hspace{-0.3cm}
\alpha e^{-\delta\tau}\sup_{X\in\Tilde{\mathcal{U}}_{0,\tau}(1,y)}\mathbb{E}\left[\frac{1}{\alpha}X_{\tau}^{\alpha}{h(Y_{\tau})}+{\frac{1}{\alpha}}A(Y_{\tau})X_{\tau}^{\alpha(1-\gamma)}\right],
\end{eqnarray}
and
\begin{eqnarray}\label{tilde.U.power}
\Tilde{\mathcal{U}}_{s,t}(x,y):=\left\{X\in\mathcal{F}^+_{t}:\sup_{\eta\in\mathcal{M}}\mathbb{E}\left[\frac{Z^{\eta}_t/B_{t}}{Z^{\eta}_s/{B_s}}X\Big|\mathcal{F}_s\right]\leq x\right\},\quad 0\leq s\leq t<\infty,
\end{eqnarray}
with $(Z^{\eta}_t)_{t\geq 0}$ defined by \eqref{z.process} and $Y_0=y$.
We aim to show that there indeed exists a unique non-negative $A^*(\cdot)\in {C}_b^{+}(\mathbb{R})$ solving equation \eqref{problem2} and then to formally prove that $\frac{1}{\alpha}A^*(y)x^{\alpha(1-\gamma)}$ coincides with the value function $V(\cdot,\cdot)$ via the verification theorem.

We next use techniques from the convex optimization theory to pose the dual problem associated with the primal problem \eqref{problem3}. For a given function $f$, let us denote its Legendre-Fenchel transform by
\begin{eqnarray}\label{Phi*}
\Phi_f(y):=\sup_{x\geq0}(f(x)-yx),\quad y\in\mathbb{R}_+.\nonumber
\end{eqnarray}
Provided that $f$ is continuous and concave with $f^{\prime}(\infty)=0$, the maximizer attaining the supremum always exists (although not necessarily unique), which is denoted by $x_f^*(y):=\arg\max_{x\geq0}(f(x)-yx)$. It follows that $\Phi_f(y)=f(x_f^*(y))-yx_f^*(y)$.

By Lemma \ref{lem2.1}, for fixed $y\in\mathbb{R}$, one can define the inverse function of $\frac{\partial}{\partial x}h_A(x,y)$ by $\mathbb{R}_+\ni x\mapsto I(x,y)\in\mathbb{R}_+$, which is a strictly decreasing function.
Let 
$\Phi_{h_A}(u,y)=\sup_{x\geq0}\{h_A(x,y)-xu\}$ be the Legendre-Fenchel transform of the concave function $h_A(x,y)$. It is well known that
\begin{eqnarray}\label{phi(y)=h}
h_A(x^*_{h_A}(u,y),y)-x^*_{h_A}(u,y)u=h_A(I(u,y),y)-I(u,y)u,
\quad u\in\mathbb{R}_+,
\end{eqnarray}
as well as when $\alpha\in(0,1)\, (\alpha\in(-\infty,0),\,\text{\text{resp.}})$
\begin{eqnarray}
\label{phi(0)}
\Phi_{h_A}(0,y)
\hspace{-0.3cm}&:=&\hspace{-0.3cm}
\lim_{u\rightarrow0+}\Phi_{h_A}(u,y)=h_A(\infty,y)=\infty\, (0,\,\text{\text{resp.}}),\\
\label{phi(infty)}
\Phi_{h_A}(\infty,y)
\hspace{-0.3cm}&:=&\hspace{-0.3cm}
\lim_{u\rightarrow\infty}\Phi_{h_A}(u,y)=h_A(0,y)=0\, (-\infty,\,\text{\text{resp.}}),
\end{eqnarray}
and 
\begin{eqnarray}\label{Phi'}
\frac{\partial}{\partial u}\Phi_{h_A}(u,y)=-x^*_{h_A}(u,y),\quad \frac{\partial^2}{\partial u^2}\Phi_{h_A}(u,y)=-\frac{\partial}{\partial u}x_{h_A}^{*}(u,y),\quad (u,y)\in\mathbb{R}_+\times\mathbb{R}.
\end{eqnarray}
In particular, the function $\Phi_{h_A}(\cdot,y)$ is a strictly decreasing, strictly convex and twice differentiable function with respect to the first argument.

Following \cite{CH05}, for fixed $y\in\mathbb{R}$, the associated dual problem to the primal problem \eqref{problem3} is defined as
\begin{align}
\label{dual.pro}
\Tilde{V}(\lambda,y)
&:=
\inf_{\eta\in\mathcal{M}}\left\{\sup_{X\in\mathcal{F}_{\tau}^+}\mathbb{E}\left[\frac{1}{\alpha}X_{\tau}^{\alpha}{h(Y_{\tau})}+\frac{1}{\alpha}A(Y_{\tau})X_{\tau}^{\alpha(1-\gamma)}-\lambda\frac{X_{\tau}Z_{\tau}^{\eta}}{B_{\tau}}\right]\right\},
\nonumber\\
&=\inf_{\eta\in\mathcal{M}}\mathbb{E}\left[\Phi_{h_A}\left(\lambda\frac{Z^{\eta}_{\tau}}{B_{\tau}},Y_{\tau}\right)\right]
\nonumber\\
&=:\inf_{\eta\in\mathcal{M}}L(\eta,\lambda),\quad \lambda\in \mathbb{R}_+,
\end{align}
where $\mathcal{F}_{\tau}^{+}$ is the set of non-negative $\mathcal{F}_{\tau}$-measurable random variables.

One can easily verify the weak duality between the primal problem \eqref{problem3} and the dual problem \eqref{dual.pro} that
\begin{eqnarray}\label{relationship}
\inf_{\lambda>0}\left\{\Tilde{V}(\lambda,y)+\lambda\right\}
\hspace{-0.3cm}&\geq&\hspace{-0.3cm}
\sup_{X\in\Tilde{\mathcal{U}}_{0,\tau}(1,y)}\mathbb{E}\left[\frac{1}{\alpha}X_{\tau}^{\alpha}h(Y_{\tau})+\frac{1}{\alpha}A(Y_{\tau})X_{\tau}^{\alpha(1-\gamma)}\right].
\end{eqnarray}
We say that there is no duality gap when the equality holds.

In the following Proposition \ref{prop2.1}, given the existence of an optimal solution to the dual problem \eqref{dual.pro}, we can present the relationship between the optimal solutions to dual problem \eqref{dual.pro} and the associated primal problem \eqref{problem3}.

\begin{prop}\label{prop2.1}
Assume that for $\alpha\in(-\infty,0)\cup(0,1)$, there exists an optimal solution to dual problem \eqref{dual.pro}. Let $(\eta^*,\lambda^*)$ be any point in $\mathcal{M}\times\mathbb{R}_+$ and
$X^*:=x_{h_A}^*(\lambda^*\frac{Z^{\eta^*}_{\tau}}{B_{\tau}},Y_{\tau})$.
If it holds that
\begin{eqnarray}\label{E=1}
X^*\in\Tilde{\mathcal{U}}_{0,\tau}(1,y)\quad\text{and}\quad\mathbb{E}\left[X^*\frac{Z^{\eta^*}_{\tau}}{B_{\tau}}\right]=1,
\end{eqnarray}
then, $X^*$ is the optimal solution to the primal problem \eqref{problem3}; $\eta^*$ is the optimal solution to the dual problem \eqref{dual.pro} and $\lambda^*=\arg\min_{\lambda>0}(\Tilde{V}(\lambda,y)+\lambda)$. In particular, there is no duality gap.
Conversely, if $\eta^*$ is the optimal solution to the dual problem \eqref{dual.pro} and $\lambda^*=\arg\min_{\lambda>0}(\Tilde{V}(\lambda,y)+\lambda)$, then \eqref{E=1} holds true and $X^*$ is the optimal solution to the primal problem \eqref{problem3}.
\end{prop}
\begin{proof}
We assume \eqref{E=1} holds true.
By Lemma \ref{lem2.1}, we know that, for any fixed $y\in\mathbb{R}$, the one-variable function $\mathbb{R}_{+}\ni x\mapsto h_{A}(x,y)$ is continuous and concave with $\frac{\partial}{\partial x}h_{A}(\infty,y)=0$. Hence, for each fixed $y\in\mathbb{R}$, $x_{h_A}^*(u,y)$ is the maximizer of the function $[0,\infty)\ni x\mapsto h_{A}(x,y)-ux$ for any $u\in\mathbb{R}_+$.
Then, for any $X\in\mathcal{F}_{\tau}^+$ and $(\eta,\lambda)\in\mathcal{M}\times\mathbb{R}_+$, using \eqref{2.8} and \eqref{phi(y)=h}, one can get
\begin{eqnarray}
\label{E<h}
h_A(X,Y_{\tau})-\lambda \frac{Z_{\tau}^{\eta}}{B_{\tau}}X\leq
h_A(x^*_{h_A}(\lambda\frac{Z_{\tau}^{\eta}}{B_{\tau}},Y_{\tau}),Y_{\tau})-\lambda x^*_{h_A}(\lambda\frac{Z_{\tau}^{\eta}}{B_{\tau}},Y_{\tau})\frac{Z_{\tau}^{\eta}}{B_{\tau}}.
\end{eqnarray}
By \eqref{dual.pro}, \eqref{E=1} and \eqref{E<h}, for each $(\eta,\lambda)\in\mathcal{M}\times\mathbb{R}_+$, it holds that
\begin{eqnarray}\label{dual.gap}
\inf_{\eta\in\mathcal{M},\lambda>0}\left\{L(\eta,\lambda)+\lambda\right\}
\hspace{-0.3cm}&=&\hspace{-0.3cm}
\inf_{\eta\in\mathcal{M},\lambda>0}\left\{\sup_{X\in\mathcal{F}_{\tau}^+}\left\{\mathbb{E}\left[h_A(X,Y_{\tau})-\lambda\frac{Z_{\tau}^{\eta}}{B_{\tau}}X\right]\right\}+\lambda\right\}
\nonumber\\
\hspace{-0.3cm}&=&\hspace{-0.3cm}
\inf_{\eta\in\mathcal{M},\lambda>0}\left\{\mathbb{E}\left[h_A(x^*_{h_A}(\lambda\frac{Z^{\eta}_{\tau}}{B_{\tau}},Y_{\tau}),Y_{\tau})
-\lambda x^*_{h_A}(\lambda\frac{Z^{\eta}_{\tau}}{B_{\tau}},Y_{\tau})\frac{Z^{\eta}_{\tau}}{B_{\tau}}\right]+\lambda\right\}
\nonumber\\
\hspace{-0.3cm}&\leq&\hspace{-0.3cm}
\mathbb{E}\left[h_A(x^*_{h_A}(\lambda^*\frac{Z^{\eta^*}_{\tau}}{B_{\tau}},Y_{\tau}),Y_{\tau})
-\lambda^* x^*_{h_A}(\lambda^*\frac{Z^{\eta^*}_{\tau}}{B_{\tau}},Y_{\tau})\frac{Z^{\eta^*}_{\tau}}{B_{\tau}}\right]+\lambda^*
\nonumber\\
\hspace{-0.3cm}&=&\hspace{-0.3cm}
\mathbb{E}\left[\frac{1}{\alpha}\left(x^*_{h_A}(\lambda^*\frac{Z^{\eta^*}_{\tau}}{B_{\tau}},Y_{\tau})\right)^{\alpha}{h(Y_{\tau})}+{\frac{1}{\alpha}}A(Y_{\tau})\left(x^*_{h_A}(\lambda^*\frac{Z^{\eta^*}_{\tau}}{B_{\tau}},Y_{\tau})\right)^{\alpha(1-\gamma)}\right]
\nonumber\\
\hspace{-0.3cm}&\leq&\hspace{-0.3cm}
\sup_{X\in\Tilde{\mathcal{U}}_{0,\tau}(1,y)}\mathbb{E}\left[\frac{1}{\alpha}X_{\tau}^{\alpha}{h(Y_{\tau})}+{\frac{1}{\alpha}}A(Y_{\tau})X_{\tau}^{\alpha(1-\gamma)}\right],
\end{eqnarray}
which together with \eqref{relationship} implies that there is no duality gap, and then, $X^*$ is the optimal solution to the primal problem \eqref{problem3}, $\eta^*$ is the optimal solution to the dual problem \eqref{dual.pro} and $\lambda^*=\arg\min_{\lambda>0}(\Tilde{V}(\lambda,y)+\lambda)$. Hence, the first claim holds true.

For the second claim, it is not hard to verify that the function $$\mathbb{R}_+\ni\lambda\mapsto\sup_{X\in\mathcal{F}_{\tau}^+}\mathbb{E}\left[h_A(X, Y_{\tau})-\lambda\frac{Z^{\eta}_{\tau}}{B_{\tau}}X\right]$$ is decreasing and convex. For any $y\in \mathbb{R}$ and $A(\cdot)\in{C}_b^{+}(\mathbb{R})$, we put
\begin{align}
\label{def.ell}
    \ell_{A,y}(x):=&h_{A}(x,y)-x\frac{\partial}{\partial x}h_{A}(x,y)
    \nonumber\\
    =&\frac{1}{\alpha}x^{\alpha}h(y)+{\frac{1}{\alpha}}A(y)x^{\alpha(1-\gamma)}
    -x\left(x^{\alpha-1}{h(y)}+A(y)(1-\gamma)x^{\alpha(1-\gamma)-1}\right)
    \nonumber\\
    =&\left(\frac{1}{\alpha}-1\right)x^{\alpha}h(y)+{\left(\frac{1}{\alpha}-(1-\gamma)\right)}A(y)x^{\alpha(1-\gamma)},\quad x\in\mathbb{R}_+.
\end{align}
By the expression of \eqref{def.ell}, it is obvious that, for any $y\in \mathbb{R}$ and $A(\cdot)\in{C}_b^{+}(\mathbb{R})$, the function $(0,\infty)\ni x\mapsto \ell_{A,y}(x)$ is bounded on any bounded interval, is increasing on $(0,\infty)$, $\lim_{x\rightarrow\infty}\ell_{A,y}(x)=\infty\,{(0,\,\text{\text{resp.}})}$, and, $\lim_{x\rightarrow0+}\ell_{A,y}(x)=0 \,{(-\infty,\,\text{resp.})}$ for $\alpha\in(0,1)\, {(\alpha\in (-\infty,0),\,\text{resp.})}$.
Therefore, by Lemma 4.2 in \cite{KL91}, Lemma \ref{lem2.1}, \eqref{E<h}, \eqref{def.ell} and the monotone convergence theorem, it holds that
\begin{align}
\label{2.36.v0}
&\lim_{\lambda\rightarrow0+}\sup_{X\in\mathcal{F}_{\tau}^+}\mathbb{E}\left[h_A(X, Y_{\tau})-\lambda\frac{XZ^{\eta}_{\tau}}{B_{\tau}}\right]
\nonumber\\
=&
\lim_{\lambda\rightarrow0+}
\mathbb{E}\left[h_A(x^*_{h_A}(\lambda\frac{Z^{\eta}_{\tau}}{B_{\tau}},Y_{\tau}),Y_{\tau})
-\lambda x^*_{h_A}(\lambda\frac{Z^{\eta}_{\tau}}{B_{\tau}},Y_{\tau})\frac{Z^{\eta}_{\tau}}{B_{\tau}}\right]
\nonumber\\
=&
\lim_{\lambda\rightarrow0+}
\mathbb{E}\left[h_A(x^*_{h_A}(\lambda\frac{Z^{\eta}_{\tau}}{B_{\tau}},Y_{\tau}),Y_{\tau})
-x^*_{h_A}(\lambda\frac{Z^{\eta}_{\tau}}{B_{\tau}},Y_{\tau})\frac{\partial}{\partial x}h_{A}(x^*_{h_A}(\lambda\frac{Z^{\eta}_{\tau}}{B_{\tau}},Y_{\tau}),Y_{\tau})\right]
\nonumber\\
=&
\lim_{\lambda\rightarrow0+}
\mathbb{E}\left[\ell_{A,Y_{\tau}}(x^*_{h_A}(\lambda\frac{Z^{\eta}_{\tau}}{B_{\tau}},Y_{\tau}))\right]
\nonumber\\
=&\infty\,{(0,\,\text{resp.})},
\end{align}
for $\alpha\in(0,1)\,{(\alpha\in(-\infty,0),\,\text{resp.})}$, because we have $x_{h_A}^*(u,y)\uparrow\infty$ as $u\downarrow0$.
Similarly, we have
\begin{align}
\label{2.37.v0}
\lim_{\lambda\rightarrow\infty}\sup_{X\in\mathcal{F}_{\tau}^+}\mathbb{E}\left[h_A(X,Y_{\tau})-\lambda\frac{XZ^{\eta}_{\tau}}{B_{\tau}}\right]
=
\lim_{\lambda\rightarrow\infty}
\mathbb{E}\left[\ell_{A,Y_{\tau}}(x^*_{h_A}(\lambda\frac{Z^{\eta}_{\tau}}{B_{\tau}},Y_{\tau}))\right]=0\,{(-\infty,\,\text{resp.})},
\end{align}
for $\alpha\in(0,1)\,{(\alpha\in(-\infty,0),\,\text{resp.})}$, since we have $x_{h_A}^*(u,y)\downarrow0$ as $u\uparrow\infty$.
In addition, note that
\begin{eqnarray}
\frac{\partial}{\partial \lambda}\left(L(\eta,\lambda)+\lambda\right)
\hspace{-0,3cm}&=&\hspace{-0.3cm}
\frac{\partial}{\partial \lambda}\left[\sup_{X\in\mathcal{F}_{\tau}^+}\mathbb{E}\left[h_A(X, Y_{\tau})-\lambda\frac{XZ^{\eta}_{\tau}}{B_{\tau}}\right]\right]+1
=
\frac{\partial}{\partial \lambda}\mathbb{E}\left[\ell_{A,Y_{\tau}}(x^*_{h_A}(\lambda\frac{Z^{\eta}_{\tau}}{B_{\tau}},Y_{\tau}))\right]+1
\nonumber\\
\hspace{-0.3cm}&=&\hspace{-0.3cm}
-\mathbb{E}\left[x^*_{h_A}\left(\lambda\frac{Z^{\eta}_{\tau}}{B_{\tau}},Y_{\tau}\right)\frac{Z^{\eta}_{\tau}}{B_{\tau}}\right]+1,\nonumber
\end{eqnarray}
which implies that there exists some $\lambda_0\in(0,\infty)$ such that $\frac{\partial}{\partial \lambda}(L(\eta,\lambda)+\lambda)<0$ for $\lambda\in(0,\lambda_0)$ and $\frac{\partial}{\partial \lambda}(L(\eta,\lambda)+\lambda)>0$ for $\lambda\in(\lambda_0,\infty)$ with $\lim_{\lambda\rightarrow\infty}\frac{\partial}{\partial \lambda}(L(\eta,\lambda)+\lambda)=1$. This, together with \eqref{2.36.v0} and \eqref{2.37.v0}, yields that $$\lim_{\lambda\rightarrow0+}(L(\eta^*,\lambda)+\lambda)=\lim_{\lambda\rightarrow\infty}(L(\eta^*,\lambda)+\lambda)=\infty,\text{ for }\alpha\in(0,1),$$
and $$\lim_{\lambda\rightarrow0+}(L(\eta^*,\lambda)+\lambda)=0,\,\,\lim_{\lambda\rightarrow\infty}(L(\eta^*,\lambda)+\lambda)=\infty,\text{ for }\alpha\in(-\infty,0).$$ 
Hence, $(L(\eta^*,\lambda)+\lambda)$ attains its infimum at some $\lambda^*\in(0,\infty)$.
Then, we have
\begin{eqnarray}
\hspace{-0.3cm}&&\hspace{-0.3cm}
\inf_{u\in\mathbb{R}_+}\left\{u\lambda+\sup_{X\in\mathcal{F}_{\tau}^+}\mathbb{E}\left[h_A(X,Y_{\tau})-u\lambda\frac{XZ^{\eta^*}_{\tau}}{B_{\tau}}\right]\right\}
\nonumber\\
\hspace{-0.3cm}&=&\hspace{-0.3cm}
\inf_{v\in\mathbb{R}_+}\left\{v+\sup_{X\in\mathcal{F}_{\tau}^+}\mathbb{E}\left[h_A(X,Y_{\tau})-v\frac{XZ^{\eta^*}_{\tau}}{B_{\tau}}\right]\right\}
=
L(\eta^*,\lambda^*)+\lambda^*.\nonumber
\end{eqnarray}
That is, the function 
\begin{align}
g(u)&:=u\lambda^*+\sup_{X\in\mathcal{F}_{\tau}^+}\mathbb{E}\left[h_A(X,Y_{\tau})-u\lambda^*\frac{XZ^{\eta^*}_{\tau}}{B_{\tau}}\right]
\nonumber\\
&=u\lambda^*+\mathbb{E}\left[
h_A(x^*_{h_A}(u\lambda^*\frac{Z^{\eta^*}_{\tau}}{B_{\tau}},Y_{\tau}),Y_{\tau})
-u\lambda^* x^*_{h_A}(u\lambda^*\frac{Z^{\eta^*}_{\tau}}{B_{\tau}},Y_{\tau})\frac{Z^{\eta^*}_{\tau}}{B_{\tau}}\right]\nonumber
\end{align}
achieves its infimum at $u=1$,
which yields that 
$$g^{\prime}(1)=\lambda^*-\lambda^*\mathbb{E}\left[x^*_{h_{A}}(\lambda^*\frac{Z^{\eta^*}_{\tau}}{B_{\tau}},Y_{\tau})\frac{Z^{\eta^*}_{\tau}}{B_{\tau}}\right]=0.$$

Then, one can find some admissible portfolio $\hat{\pi}$ that finances $x^*_{h_{A}}(\lambda^*\frac{Z^{\eta^*}_{\tau}}{B_{\tau}},Y_{\tau})$ by an application of martingale representation theorem to the martingale $$\frac{X^{1,y,\hat{\pi}}_tZ^{{\eta}^*}_t}{B_t}:=\mathbb{E}\left[x^*_{h_{A}}(\lambda^*\frac{Z^{\eta^*}_{\tau}}{B_{\tau}},Y_{\tau})\frac{Z^{\eta^*}_{\tau}}{B_{\tau}}\Bigg|\mathcal{F}_{t}\right],\quad t\in[0,\tau],$$
with $X^{1,y,\hat{\pi}}_0=1$ and $X^{1,y,\hat{\pi}}_\tau=x^*_{h_{A}}(\lambda^*\frac{Z^{\eta^*}_{\tau}}{B_{\tau}},Y_{\tau})$ (see Section 6 in \cite{KL91} for more details).
Furthermore, from \eqref{supermaringale}, we know $X^{1,y,\hat{\pi}}\frac{Z^{\eta}}{B}$ is a $\mathbb{P}$-supermartingale for any $\eta\in\mathcal{M}$. Hence, it holds that
$$\mathbb{E}\left[x^*_{h_{A}}(\lambda^*\frac{Z^{\eta^*}_{\tau}}{B_{\tau}},Y_{\tau})\frac{Z^{\eta}_{\tau}}{B_{\tau}}\right]=\mathbb{E}\left[X^{1,y,\hat{\pi}}_{\tau}\frac{Z^{\eta}_{\tau}}{B_{\tau}}\right]\leq 1,\quad \eta\in\mathcal{M},$$
which combined with the arbitrary of $\eta$ implies $x^*_{h_{A}}(\lambda^*\frac{Z^{\eta^*}_{\tau}}{B_{\tau}},Y_{\tau})\in\Tilde{\mathcal{U}}_{0,\tau}(1,y)$.
Consequently, \eqref{E=1} holds. Using the first claim, one knows that $X^*$ is the optimal solution to the primal problem \eqref{problem3}.
The proof is complete.
\end{proof}

To prove the existence of the optimal solution to dual problem \eqref{dual.pro}, let us introduce the Dol\'{e}ans-Dade exponential 
\begin{eqnarray}\label{DD.exp}
\mathcal{E}_t^{W}(\eta):=\exp\left\{\int_0^t\eta_sdW_s-\frac{1}{2}\int_0^t\eta^2_sds\right\},\quad t\in[0,\tau].
\end{eqnarray}
The following Propositions \ref{existence.1} and \ref{existence.2} address the existence  of the optimal solution to the dual problem \eqref{dual.pro} for two separate cases $\alpha\in(0,1)$ and $\alpha\in(-\infty,0)$, respectively.
\begin{prop}
\label{existence.1}
For $\alpha\in(0,1)$, there exists an optimal solution to the dual problem \eqref{dual.pro}.
\end{prop}
\begin{proof}
Denote by $\mathcal{F}^{Y}_{\tau}$ the smallest sigma-field generated by $(Y_{t})_{0\leq t\leq \tau}$.
We first verify that, for fixed $\lambda\in\mathbb{R}_+$, $L(\cdot,\lambda)$ defined by \eqref{dual.pro} is a convex functional on $\mathcal{M}$.
Notice that the Arrow-Pratt measure of relative risk aversion 
\begin{align}\label{APmrra}
    -\frac{x\frac{\partial^2}{\partial x^2}h_A(x,y)}{\frac{\partial}{\partial x}h_A(x,y)}
    &=-\frac{(\alpha-1)x^{\alpha-1}h(y)+(1-\gamma)(\alpha(1-\gamma)-1)x^{\alpha(1-\gamma)-1}A(y)}{x^{\alpha-1}h(y)+(1-\gamma)x^{\alpha(1-\gamma)-1}A(y)}\nonumber\\
    &\leq1, \quad (x,y)\in\mathbb{R}_+\times\mathbb{R}.
\end{align}
Then, by Lemma 12.6 of \cite{KL91}, one knows the function $z\mapsto \Phi_{h_A}(e^{z},y)$ is convex on $\mathbb{R}$ for any fixed $y\in\mathbb{R}$. This, together with the convexity of the Euclidean norm in $\mathbb{R}$, the decrease of $\Phi_{h_A}$ and Jensen's inequality, implies
\begin{eqnarray}
\hspace{-0.3cm}&&\hspace{-0.3cm}
L(\omega_1\eta_1+\omega_2\eta_2,\lambda)
\nonumber\\
\hspace{-0.3cm}&\leq&\hspace{-0.3cm}
\mathbb{E}\left[\Phi_{h_A}\left(\frac{\lambda}{B_{\tau}}\mathcal{E}^{W_{1}}_{\tau}(-\theta)(\mathcal{E}^{W_{2}}_{\tau}(\eta_1))^{\omega_1}(\mathcal{E}^{W_{2}}_{\tau}(\eta_2))^{\omega_2},Y_{\tau}\right)\right]
\nonumber\\
\hspace{-0.3cm}&=&\hspace{-0.3cm}
\mathbb{E}\left[\mathbb{E}\left[\left.\Phi_{h_A}\left(\frac{\lambda}{B_{\tau}}\mathcal{E}^{W_{1}}_{\tau}(-\theta)(\mathcal{E}^{W_{2}}_{\tau}(\eta_1))^{\omega_1}(\mathcal{E}^{W_{2}}_{\tau}(\eta_2))^{\omega_2},Y_{\tau}\right)\right|\mathcal{F}^{Y}_{\tau}\right]\right]
\nonumber\\
\hspace{-0.3cm}&\leq&\hspace{-0.3cm}
\mathbb{E}\left[\omega_1\mathbb{E}\left[\left.\Phi_{h_A}\left(\frac{\lambda}{B_{\tau}}\mathcal{E}^{W_1}_{\tau}(-\theta)\mathcal{E}^{W_2}_{\tau}(\eta_1),Y_{\tau}\right)\right|\mathcal{F}^{Y}_{\tau}\right]\right.
\nonumber\\
\hspace{-0.3cm}&&\hspace{0.3cm}
\left.+\omega_2\mathbb{E}\left[\left.\Phi_{h_A}\left(\frac{\lambda}{B_{\tau}}\mathcal{E}^{W_1}_{\tau}(-\theta)\mathcal{E}^{W_2}_{\tau}(\eta_2),Y_{\tau}\right)\right|\mathcal{F}^{Y}_{\tau}\right]\right]
\nonumber\\
\hspace{-0.3cm}&=&\hspace{-0.3cm}
\omega_1 L(\eta_1,\lambda)+\omega_2 L(\eta_2,\lambda),
\end{eqnarray}
for any $\eta_1,\eta_2\in\mathcal{M}$ and $\omega_1,\omega_2\geq0$ with $\omega_1+\omega_2=1$. Hence, for fixed $\lambda\in\mathbb{R}_+$, $L(\cdot,\lambda)$ is a convex functional on $\mathcal{M}$.
Furthermore, one can employ similar methods as those adopted in the proof of Theorem 12.3 in \cite{KL91} to prove the existence of optimal solution to dual problem \eqref{dual.pro} when $\alpha\in(0,1)$ since the remaining assumptions of this theorem are fulfilled.
\end{proof}

The next result handles the more tricky case when $\alpha\in(-\infty,0)$. In this case, the arguments in \cite{KL91} (see, Lemma 12.6 and Theorem 12.3 of \cite{KL91}) are no longer applicable because \eqref{APmrra} does not hold. In fact, $z\mapsto \Phi_{h_A}(e^{z},y)$ is concave on $\mathbb{R}$. We note that \cite{La11} investigated this special case and established the existence of the optimal solution to the dual problem when the utility function is of form $U(x)=\frac{1}{\alpha}x^{\alpha}$ for $\alpha<0$, by connecting the dual problem to a fictitious market through a change-of-measure argument. The use of this change-of-measure argument is possible because their dual functional allows for the separation of the variable that needs to be optimized. In our context, however, due to the form of modified utility function $h_A(x,y)=\frac{1}{\alpha}x^{\alpha}h(y)+\frac{1}{\alpha}x^{\alpha(1-\gamma)}A(y)$, it becomes infeasible for our dual functional to separate the variables. Nevertheless, we show below that it is not necessary to change the measure, and we can still finish the task by establishing the relationship to another artificial optimization problem.

\begin{prop}
\label{existence.2}
For $\alpha\in(-\infty,0)$, there exists an optimal solution to the dual problem \eqref{dual.pro}.
\end{prop}
\begin{proof}
We consider the artificial risky asset $\Tilde{S}=(\Tilde{S}_t)_{t\geq0}$ given by
$$d\Tilde{S}_t=\Tilde{S}_tdW_{2t},\quad t\in[0,\tau],$$
with $\Tilde{S}_0=1,$ and the artificial risk-free asset with interest rate 0. A trading strategy $\eta=(\eta_t)_{t\geq0}$ is a predictable process representing the admissible portfolio fraction invested in the risky asset $\Tilde{S}$ at time $t$. The resulting wealth process $\Tilde{X}$ satisfies $$d\Tilde{X}_t=\Tilde{X}_t\eta_tdW_{2t},\quad t\in[0,\tau],$$
with $\Tilde{X}_0=1.$ Recall that the Dol\'{e}ans-Dade exponential $\mathcal{E}_t^W(\cdot)$ is defined by \eqref{DD.exp}. We note that $\mathcal{E}^{W_2}_t(\eta)=\Tilde{X}_t$, $t\in[0,\tau]$. Furthermore, we define the set of wealth processes by
\begin{eqnarray}
\mathcal{V}(1,y)\hspace{-0.3cm}&:=&\hspace{-0.3cm}
\left\{\Tilde{X}:\Tilde{X}_t = 1+\int_0^{t}\Tilde{X}_s\eta_sdW_{2s}>0\text{ for } t\in[0,\tau],\, \eta\text{ is predictable and satisfies}\right.
\nonumber\\
\hspace{-0.3cm}&&\hspace{0cm}
\left.\mathbb{E}\left[\int_0^{\tau}\eta^2_sds\right]<\infty,\text{ and, }\mathbb{E}\left[\left(-\Phi_{h_A}\left(\frac{\lambda}{B_{\tau}}\mathcal{E}^{W_1}_{\tau}(-\theta)\Tilde{X}_{\tau},Y_{\tau}\right)\right)_-\right]<\infty\right\},\nonumber
\end{eqnarray}
with $Y_0=y\in\mathbb{R}$.
When $\alpha\in(-\infty,0)$, it follows from \eqref{phi(0)} and \eqref{phi(infty)} that $\Phi_{h_A}(x,y)\leq0$ for any $(x,y)\in\mathbb{R}_+\times\mathbb{R}$, and hence
$$\mathbb{E}\left[\left(-\Phi_{h_A}\left(\frac{\lambda}{B_{\tau}}\mathcal{E}^{W_1}_{\tau}(-\theta)\Tilde{X}_{\tau},Y_{\tau}\right)\right)_-\right]<\infty$$ is readily satisfied.
Therefore, it can be seen that if $\eta\in\mathcal{M}$ (recall that 
$\mathcal{M}$ is the set of progressively measurable processes $(\eta_t)_{t\geq0}$
such that $\mathbb{E}[\int_0^{\tau}\eta_s^2ds]<\infty$ and the local martingale by
\eqref{Z.process} is a true martingale), then the corresponding  wealth process $\Tilde{X}_t=\mathcal{E}_t^{W_2}(\eta)=1+\int_0^t\Tilde{X}_s\eta_sdW_{2s}\in\mathcal{V}(1,y)$; and vice versa.
Then, we can rewrite the dual problem \eqref{dual.pro} as
\begin{eqnarray}
\label{connection}
\Tilde{V}(\lambda,y)
\hspace{-0.3cm}&=&\hspace{-0.3cm}
\inf_{\eta\in\mathcal{M}}\mathbb{E}\left[\Phi_{h_A}\left(\lambda\frac{Z^{\eta}_{\tau}}{B_{\tau}},Y_{\tau}\right)\right]
\nonumber\\
\hspace{-0.3cm}&=&\hspace{-0.3cm}
-\sup_{\eta\in\mathcal{M}}\mathbb{E}\left[-\Phi_{h_A}\left(\frac{\lambda}{B_{\tau}}\mathcal{E}^{W_1}_{\tau}(-\theta)\mathcal{E}^{W_2}_{\tau}(\eta),Y_{\tau}\right)\right]
\nonumber\\
\hspace{-0.3cm}&=&\hspace{-0.3cm}
-\sup_{\Tilde{X}\in\mathcal{V}(1,y)}\mathbb{E}\left[-\Phi_{h_A}\left(\frac{\lambda}{B_{\tau}}\mathcal{E}^{W_1}_{\tau}(-\theta)\Tilde{X}_{\tau},Y_{\tau}\right)\right]
\nonumber\\
\hspace{-0.3cm}&=:&\hspace{-0.3cm}
-v(1,y),\quad y\in\mathbb{R}.
\end{eqnarray}
From \eqref{2.12}, \eqref{relationship} and the fact that $\Phi_{h_A}(\cdot,\cdot)\leq0$ when $\alpha\in(-\infty,0)$, it holds that
\begin{eqnarray}
\lambda\geq \Tilde{V}(\lambda,y)+\lambda \geq \sup_{X\in\Tilde{\mathcal{U}}_{0,\tau}(1,y)}\mathbb{E}\left[\frac{1}{\alpha}X_{\tau}^{\alpha}h(Y_{\tau})+\frac{1}{\alpha}A(Y_{\tau})X_{\tau}^{\alpha(1-\gamma)}\right]>-\infty.\nonumber
\end{eqnarray}
Consequently, one obtains that the artificial primal problem 
\begin{eqnarray}\label{ar.primal}
v(1,y)=\sup_{\Tilde{X}\in\mathcal{V}(1,y)}\mathbb{E}\left[-\Phi_{h_A}\left(\frac{\lambda}{B_{\tau}}\mathcal{E}^{W_1}_{\tau}(-\theta)\Tilde{X}_{\tau},Y_{\tau}\right)\right]<\infty,\quad y\in\mathbb{R}.
\end{eqnarray}
Furthermore, the artificial dual problem associated to the artificial primal problem \eqref{ar.primal} is defined by
\begin{eqnarray}\label{ar.dual}
\Tilde{v}(z,y)
\hspace{-0.3cm}&:=&\hspace{-0.3cm}
\inf_{\nu\in\mathcal{M}}\left\{\sup_{\Tilde{X}\in\mathcal{F}^{+}_{\tau}}\mathbb{E}\left[-\Phi_{h_A}\left(\frac{\lambda}{B_{\tau}}\mathcal{E}_{\tau}^{W_1}(-\theta)\Tilde{X},Y_{\tau}\right)-z\mathcal{E}_{\tau}^{W_1}(\nu)\Tilde{X}\right]\right\}
\nonumber\\
\hspace{-0.3cm}&=&\hspace{-0.3cm}
\inf_{\nu\in\mathcal{M}}\mathbb{E}\left[-\Phi_{h_A}\left(\frac{\partial}{\partial x}h_A\left(z\frac{B_{\tau}\mathcal{E}_{\tau}^{W_1}(\nu)}{\lambda\mathcal{E}_{\tau}^{W_1}(-\theta)},Y_{\tau}\right),Y_{\tau}\right)-\frac{\partial}{\partial x}h_A\left(z\frac{B_{\tau}\mathcal{E}_{\tau}^{W_1}(\nu)}{\lambda\mathcal{E}_{\tau}^{W_1}(-\theta)},Y_{\tau}\right)z\frac{B_{\tau}\mathcal{E}_{\tau}^{W_1}(\nu)}{\lambda\mathcal{E}_{\tau}^{W_1}(-\theta)}\right]
\nonumber\\
\hspace{-0.3cm}&=&\hspace{-0.3cm}
\inf_{\nu\in\mathcal{M}}\mathbb{E}\left[-h_A\left(z\frac{B_{\tau}}{\lambda\mathcal{E}_{\tau}^{W_1}(-\theta)}\mathcal{E}_{\tau}^{W_1}(\nu),Y_{\tau}\right)\right],\quad (z,y)\in\mathbb{R}_+\times\mathbb{R},
\end{eqnarray}
where, in the last equality, we have used the property (4.5) in \cite{KL91}. 
It follows from Lemma \ref{lem2.1} and \eqref{M.def} that 
\begin{eqnarray}
0<\Tilde{v}(z,y)
\hspace{-0.3cm}&\leq&\hspace{-0.3cm}
\mathbb{E}\left[-h_A\left(\frac{zB_{\tau}}{\lambda\mathcal{E}_{\tau}^{W_1}(-\theta)}\mathcal{E}_{\tau}^{W_1}(\nu),Y_{\tau}\right)\right]\Bigg|_{\nu\equiv0}
\nonumber\\
\hspace{-0.3cm}&\leq&\hspace{-0.3cm}
-\kappa_2\left(1+\left(\frac{zB_{\tau}}{\lambda}\right)^{\rho_2}\mathbb{E}\left[\exp\left(\rho_2\int_0^{\tau}\theta(Y_t)dW_{1t}+\rho_2\int_0^{\tau}\theta^2(Y_t)dt\right)\right]\right)
\nonumber\\
\hspace{-0.3cm}&<&\hspace{-0.3cm}
\infty,\quad(z,y)\in\mathbb{R}_+\times\mathbb{R}.\nonumber
\end{eqnarray}

Due to the strict decreasing property and convexity of $\mathbb{R}_{+}\ni x\mapsto-h_{A}(e^{x},y)\in\mathbb{R}_{+}$, one can verify that the functional $\mathbb{E}\big[-h_A\big(zB_{\tau}\mathcal{E}_{\tau}^{W_1}(\nu)/\lambda\mathcal{E}_{\tau}^{W_1}(-\theta),Y_{\tau}\big)\big]$ is convex with respect to $\nu$ on $\mathcal{M}$, and then establish the existence of optimal solution to the artificial dual problem \eqref{ar.dual} by employing a similar argument as that of Proposition \ref{existence.1}. Following the arguments of Proposition \ref{prop2.1}, we know that the existence of optimal solution to the artificial dual problem \eqref{ar.dual} implies the existence of optimal solution to the artificial primal problem \eqref{ar.primal}, and then, by \eqref{connection}, implies the existence of optimal solution to the dual problem \eqref{dual.pro}.
\end{proof}

In the next result, we show that there indeed exists a unique $A^*(\cdot)$ solving equation \eqref{A*.def} and we give the upper and lower bounds for the fixed-point $A^*(\cdot)$.

\begin{prop}\label{fixed.power}
Recall that the functional $\Psi: {C}_b^{+}(\mathbb{R})\mapsto {C}_b^{+}(\mathbb{R})$ is given by \eqref{problem3}.
Then $\Psi$ is a contraction on the metric space $({C}_b^{+}(\mathbb{R}),d)$ with the metric $d$ defined by $d(x,y):=\sup_{t\in\mathbb{R}}|x(t)-y(t)|,$  $x,y\in {C}_b^{+}(\mathbb{R})$, and $\Psi$ admits a unique fixed-point $A^*(\cdot)$ such that $A^*(y)=\Psi(A^*;y)$, for all $y\in\mathbb{R}$.
Moreover, the unique fixed-point $A^*(\cdot)$ of $\Psi$ satisfies, when $\alpha\in(0,1)$,
    $$\frac{me^{(r\alpha-\delta)\tau}}{(1-e^{-(\delta-r\alpha(1-\gamma))\tau})}\leq A^*(y)\leq \frac{e^{(\zeta(\alpha)-\delta)\tau}}{(1-e^{-(\delta-\zeta(\alpha(1-\gamma)))\tau})},\quad y\in\mathbb{R};$$
    and when $\alpha\in(-\infty,0)$,
    $$\frac{m e^{(\zeta(\alpha)-\delta)\tau}}{(1-e^{-(\delta-\zeta(\alpha(1-\gamma)))\tau})}\leq A^*(y)\leq  
 \frac{e^{(r\alpha-\delta)\tau}}{(1-e^{-(\delta-r\alpha(1-\gamma))\tau})},\quad y\in\mathbb{R}.$$
\end{prop}
\begin{proof}
Because the arguments for cases $\alpha\in(0,1)$ and $\alpha\in(-\infty,0)$ are similar, we shall only present the proof when $\alpha\in(0,1)$.
It is not difficult to verify that the metric space $(C^+_b(\mathbb{R}),d)$ is complete. Consider any $A_1,A_2\in C^+_b(\mathbb{R})$. By Proposition \ref{prop2.1}, there exists an optimizer $X^*_1\in\Tilde{\mathcal{U}}_{0,\tau}(1,y)$ such that 
\begin{eqnarray}
\Psi(A_1) 
\hspace{-0.3cm}&=&\hspace{-0.3cm}
\alpha\sup_{X\in\Tilde{\mathcal{U}}_{0,\tau}(1,y)}\mathbb{E}\left[\frac{1}{\alpha}e^{-\delta\tau}X^{\alpha}_{\tau}h(Y_{\tau})+\frac{1}{\alpha}A_1(Y_{\tau})e^{-\delta\tau}X_{\tau}^{\alpha(1-\gamma)}\right]
\nonumber\\
\hspace{-0.3cm}&=&\hspace{-0.3cm}
\mathbb{E}\left[e^{-\delta\tau}(X^*_{1\tau})^{\alpha}h(Y_{\tau})+A_1(Y_{\tau})e^{-\delta\tau}(X^*_{1\tau})^{\alpha(1-\gamma)}\right],\nonumber
\end{eqnarray}
and, by the definition of metric $d$, there exists a constant $y^*\in\mathbb{R}$ such that $d(\Psi(A_1),\Psi(A_2))=|\Psi(A_1;y^*)-\Psi(A_2;y^*)|$. Without loss of generality, we assume that $\Psi(A_1;y^*)\geq\Psi(A_2;y^*)$. Then, it holds that
\begin{eqnarray}
d(\Psi(A_1),\Psi(A_2))
\hspace{-0.3cm}&=&\hspace{-0.3cm}
\left[\mathbb{E}\left[e^{-\delta\tau}(X^*_{1\tau})^{\alpha}h(Y_{\tau})+A_1(Y_{\tau})e^{-\delta\tau}(X^*_{1\tau})^{\alpha(1-\gamma)}\right]\right.
\nonumber\\
\hspace{-0.3cm}&&\hspace{-0.3cm}
\left.-\sup_{X\in\Tilde{\mathcal{U}}_{0,\tau}(1,y)}\mathbb{E}\left[e^{-\delta\tau}X^{\alpha}_{\tau}h(Y_{\tau})+A_2(Y_{\tau})e^{-\delta\tau}X_{\tau}^{\alpha(1-\gamma)}\right]\right]\Bigg|_{y=y^*}
\nonumber\\
\hspace{-0.3cm}&\leq&\hspace{-0.3cm}
\left[\mathbb{E}\left[e^{-\delta\tau}(X^*_{1\tau})^{\alpha}h(Y_{\tau})+A_1(Y_{\tau})e^{-\delta\tau}(X^*_{1\tau})^{\alpha(1-\gamma)}\right]\right.
\nonumber\\
\hspace{-0.3cm}&&\hspace{-0.3cm}
\left.-\mathbb{E}\left[e^{-\delta\tau}(X_{1\tau}^*)^{\alpha}h(Y_{\tau})+A_2(Y_{\tau})e^{-\delta\tau}(X^*_{1\tau})^{\alpha(1-\gamma)}\right]\right]\Big|_{y=y^*}
\nonumber\\
\hspace{-0.3cm}&\leq&\hspace{-0.3cm}
e^{-\delta\tau}\sup_{X\in\Tilde{\mathcal{U}}_{0,\tau}(1,y)}\mathbb{E}\left[X_{\tau}^{\alpha(1-\gamma)}\right]\Big|_{y=y^*}\times d(A_1,A_2)
\nonumber\\
\hspace{-0.3cm}&\leq&\hspace{-0.3cm}
e^{-(\delta-\zeta(\alpha(1-\gamma)))\tau}d(A_1,A_2),\nonumber
\end{eqnarray}
where in the last inequality, we have used the fact that
$$\sup_{X\in\Tilde{\mathcal{U}}_{0,\tau}(1,y)}\mathbb{E}\left[X_{\tau}^{\alpha(1-\gamma)}\right]=e^{\zeta(\alpha(1-\gamma))\tau},\quad y\in\mathbb{R}.$$
Due to the standing Assumption \ref{ass1}, the existence and uniqueness of a fixed-point $A^*\in C^+_b(\mathbb{R})$ for $\Psi$ immediately
follow from the Banach contraction theorem.

Next, define $\underline{A}:=\inf_{y\in\mathbb{R}}A(y)$ and $\overline{A}:=\sup_{y\in\mathbb{R}}A(y)$ for any $A(\cdot)\in C_b^+(\mathbb{R})$.
Noting that $X=e^{r\tau}$ is admissible to the problem \eqref{problem}, we readily obtain that 
\begin{eqnarray}
A^*(y)=\Psi(A^*;y)\geq \frac{1}{\alpha}me^{(r\alpha-\delta)\tau}+\frac{1}{\alpha}\underline{A}^*e^{-(\delta-r\alpha(1-\gamma))\tau},\quad y\in\mathbb{R}.\nonumber
\end{eqnarray}
Due to the arbitrariness of $y$, it holds that
$$\underline{A}^*\geq \frac{1}{\alpha}me^{(r\alpha-\delta)\tau}+\frac{1}{\alpha}\underline{A}^*e^{-(\delta-r\alpha(1-\gamma))\tau},$$
which yields the lower bound of $A^*$. For the upper bound, using \eqref{inv.pro}, we have
\begin{eqnarray}
A^*(y)=\Psi(A^*;y)
\hspace{-0.3cm}&=&\hspace{-0.3cm}
\alpha\sup_{X\in\Tilde{\mathcal{U}}_{0,\tau}(1,y)}\mathbb{E}\left[\frac{1}{\alpha}e^{-\delta\tau}X_{\tau}^{\alpha}{h(Y_{\tau})}+\frac{1}{\alpha}{A^*(Y_{\tau})}e^{-\delta\tau}X_{\tau}^{\alpha(1-\gamma)}\right]
\nonumber\\
\hspace{-0.3cm}&\leq&\hspace{-0.3cm}
\frac{1}{\alpha}e^{-\delta\tau}\sup_{X\in\Tilde{\mathcal{U}}_{0,\tau}(1,y)}\mathbb{E}\left[X_{\tau}^{\alpha}\right]+\frac{1}{\alpha}\overline{A}^*e^{-\delta\tau}\sup_{X\in\Tilde{\mathcal{U}}_{0,\tau}(1,y)}\mathbb{E}\left[X_{\tau}^{\alpha(1-\gamma)}\right]
\nonumber\\
\hspace{-0.3cm}&=&\hspace{-0.3cm}
\frac{1}{\alpha}e^{(\zeta(\alpha)-\delta)\tau}+\frac{1}{\alpha}\overline{A}^*e^{-(\delta-\zeta(\alpha(1-\gamma)))\tau},\quad y\in\mathbb{R}.\nonumber
\end{eqnarray}
Due to the arbitrariness of $y$, it holds that $$\overline{A}^*\leq \frac{1}{\alpha}e^{(\zeta(\alpha)-\delta)\tau}+\overline{A}^*e^{-(\delta-\zeta(\alpha(1-\gamma)))\tau},$$which, together with  $\delta>\zeta(\alpha(1-\gamma))\vee0$ from Assumption \ref{ass1}, gives the upper bound.
\end{proof}

With the preparations, our next goal is to identify the solution (i.e., the optimal trading strategy as well as the optimal value function) to the problem \eqref{problem}. 
 To this end, we first give a representation of the wealth $X$ in the next Lemma \ref{prop.euquiv}. Its proof is similar to Theorem 2.3 in \cite{CH05}, and is hence omitted.
\begin{lem}
\label{prop.euquiv}
Let $\hat{\eta}\in\mathcal{M}$. The following statements are equivalent.
\begin{itemize}
    \item [\rm{(i)}] $\hat{X}\in\Tilde{\mathcal{U}}_{0,\tau}(x,y)$ and $\mathbb{E}\left[\frac{\hat{X}_{\tau}Z^{\hat{\eta}}_{\tau}}{B_{\tau}}\right]=x.$
    \item[\rm{(ii)}] There exists an $X^{x,y,\pi}\in\mathcal{U}_0(x,y)$ such that  $X^{x,y,\pi}_{\tau}\equiv\hat{X}$ and $\frac{X^{x,y,\pi}}{B}$ is a $\mathbb{P}^{\hat{\eta}}$-martingale with the representation
    $$\frac{X^{x,y,\pi}_t}{B_t}=x+\int_0^t\psi_{s}dW^{\hat{\eta}}_{1s},\quad t\in[0,\tau],$$ where $\psi$ is a progressively measurable process with $\int_0^{\tau}\psi_u^2du<\infty$ almost surely.
\end{itemize}
\end{lem}

For the later use, we introduce the following notations. Write $(\mathcal{M}_n)_{n\geq1}$ as a series of sets of progressively measurable processes $(\eta_{nt})_{t\in[T_{n-1},T_n)}$ such that $\mathbb{E}[\int_{T_{n-1}}^{T_n}\eta_{nt}^2dt]<\infty$ with $\mathcal{M}_1=\mathcal{M}$, and define the processes 
$(Z^{\eta_n}_{T_n}/Z^{\eta_{n-1}}_{T_{n-1}})_{n\geq1}$ by 
\begin{eqnarray}\label{def.Z.eta}
\frac{Z^{\eta_n}_{T_n}}{Z^{\eta_{n-1}}_{T_{n-1}}}:=\exp\left(-\int_{T_{n-1}}^{T_n}[\theta(Y_s)dW_{1s}-\eta_{ns}dW_{2s}]-\frac{1}{2}\int_{T_{n-1}}^{T_n}[\theta^2(Y_s)+\eta^2_{ns}]ds\right),
\end{eqnarray}
with $\eta_0\equiv0$ and $Z^{\eta_0}_{T_0}=1$.

We are now ready to state the main result in this section.
\begin{thm}\label{thm3.1}
Let $(\eta_n^*,\lambda_n^*)_{n\geq1}$ 
be the optimal solutions to a family of dual problems as follows
\begin{eqnarray}
\inf_{\eta_n\in\mathcal{M}_n,\lambda_n>0}(L_n(\eta_n,\lambda_n)+\lambda_n)
\hspace{-0.3cm}&:=&\hspace{-0.3cm}
\inf_{\eta_n\in\mathcal{M}_n,\lambda_n>0}\left\{\sup_{X\in\mathcal{F}_{T_{n}}^+}\left\{\mathbb{E}\left[\frac{1}{\alpha}X^{\alpha}{h(Y_{T_n})}+{\frac{1}{\alpha}}A(Y_{T_n})X^{\alpha(1-\gamma)}\right.\right.\right.
\nonumber\\
\hspace{-0.3cm}&&\hspace{2.3cm}
\left.\left.-\lambda_n\left.X\frac{Z_{T_n}^{\eta_n}/B_{T_n}}{Z^{\eta_{n-1}}_{T_{n-1}}/B_{T_{n-1}}}\bigg|\mathcal{F}_{T_{n-1}}\right]\right\}+\lambda_n\right\},\nonumber
\end{eqnarray}
and 
\begin{eqnarray}\label{X^*}
X^*_{T_n}:=X^*_{T_{n-1}}x^*_{h_{A^*}}\left(\lambda_n^*\frac{Z^{\eta_{n}^*}_{T_n}/B_{T_n}}{Z^{\eta_{n-1}^*}_{T_{n-1}}/B_{T_{n-1}}},Y_{T_{n}}\right),\quad n\geq 1,
\end{eqnarray}
and $X^*_{T_0}=x$, where the function $x^*_{h_{A^*}}(\cdot,\cdot)$ is defined by \eqref{phi(y)=h} with $A^*(\cdot)\in C^+_b(\mathbb{R})$ being the unique fixed-point of the function $\Psi$ defined in \eqref{problem3}, and 
the process $Z^{\eta_n^*}_{T_n}$ is given by \eqref{def.Z.eta} with $\eta_n$ replaced by $\eta^*_n$ for all $n\geq1$. 
Then, the value function of the problem \eqref{problem} is given by
\begin{eqnarray}
V(x,y)={\frac{1}{\alpha}}A^*(y)x^{\alpha(1-\gamma)},\quad(x,y)\in\mathbb{R}_+\times\mathbb{R}.\nonumber
\end{eqnarray}
The optimal wealth process $X^*$ at time $T_n$ is given by \eqref{X^*}. In particular, the optimal proportion of wealth invested in the risky asset is given by 
\begin{eqnarray}
\label{pi^*}
\pi^*_t=\frac{\psi^*_{nt}B_t}{\sigma(Y_t)},\quad T_{n}\leq t<T_{n+1},
\end{eqnarray}
for some ${\mathcal{F}_t}$-adapted process $\psi^*_n$ satisfying $\int_{T_n}^{T_{n+1}}{\psi^{*}_{nt}}^2dt<\infty$ almost surely. 
\end{thm}
\begin{proof}
For any admissible wealth process $X\in\mathcal{U}_0(x,y)$, define a discrete-time stochastic process $D=(D_n)_{n=0,1,2,...}$ by
$$D_n:=\sum_{i=1}^n\frac{1}{\alpha}e^{-\delta T_i}\left(\frac{X_{T_i}}{X_{T_{i-1}}^{\gamma}}\right)^{\alpha}h(Y_{T_i})+{\frac{1}{\alpha}}A^*(Y_{T_n})e^{-\delta T_n}X_{T_n}^{\alpha(1-\gamma)}.$$
Then
\begin{eqnarray}\label{Dn+1}
D_{n+1}
\hspace{-0.3cm}&=&\hspace{-0.3cm}
D_{n}+e^{-\delta T_{n+1}}\frac{1}{\alpha}\bigg(\frac{X_{T_{n+1}}}{X_{T_{n}}^{\gamma}}\bigg)^{\alpha}h(Y_{T_{n+1}})-{\frac{1}{\alpha}}A^*(Y_{T_n})e^{-\delta T_n}X_{T_n}^{\alpha(1-\gamma)}+{\frac{1}{\alpha}}A^*(Y_{T_{n+1}})e^{-\delta T_{n+1}}X_{T_{n+1}}^{\alpha(1-\gamma)}
\nonumber\\
\hspace{-0.3cm}&=&\hspace{-0.3cm}
D_n+e^{-\delta T_n}\left[e^{-\delta\tau}\left(\frac{1}{\alpha}\bigg(\frac{X_{T_{n+1}}}{X^{\gamma}_{T_{n}}}\bigg)^{\alpha}h(Y_{T_{n+1}})+{\frac{1}{\alpha}}A^*(Y_{T_{n+1}})X^{\alpha(1-\gamma)}_{T_{n+1}}\right)-{\frac{1}{\alpha}}A^*(Y_{T_n})X^{\alpha(1-\gamma)}_{T_n}\right].
\nonumber\\
\end{eqnarray}
Taking the conditional expectation on both sides of \eqref{Dn+1}, one gets
\begin{eqnarray}\label{3.4}
\mathbb{E}[D_{n+1}|\mathcal{F}_{T_n}]
\hspace{-0.3cm}&=&\hspace{-0.3cm}
D_n+e^{-\delta T_n}X^{\alpha(1-\gamma)}_{T_n}\left[e^{-\delta\tau}\mathbb{E}\left[\frac{1}{\alpha}\bigg(\frac{X_{T_{n+1}}}{X_{T_n}}\bigg)^{\alpha}h(Y_{T_{n+1}})\right.\right.
\nonumber\\
\hspace{-0.3cm}&&\hspace{3cm}
\left.\left.+{{\frac{1}{\alpha}}A^*(Y_{T_{n+1}})}\bigg(\frac{X_{T_{n+1}}}{X_{T_n}}\bigg)^{\alpha(1-\gamma)}\bigg|\mathcal{F}_{T_n}\right]-{\frac{1}{\alpha}}A^*(Y_{T_n})\right].
\end{eqnarray}
In view of $X\in\mathcal{U}_0(x,y)$ and \eqref{supermaringale}, for $\eta_{n+1}\in\mathcal{M}_{n+1}$, one knows that $XZ^{\eta_{n+1}}/B$ is a $\mathbb{P}$-supermartingale over $[T_n,T_{n+1}]$, which yields that $$\mathbb{E}\left[X_{T_{n+1}}\frac{Z^{\eta_{n+1}}_{T_{n+1}}}{B_{T_{n+1}}}\bigg|\mathcal{F}_{T_n}\right]\leq\mathbb{E}\left[X_{T_n}\frac{Z^{\eta_n}_{T_n}}{B_{T_n}}\right].$$
Furthermore, it holds that
\begin{eqnarray}\label{3.5}
\hspace{-0.3cm}&&\hspace{-0.3cm}
e^{-\delta\tau}\mathbb{E}\left[\frac{1}{\alpha}\bigg(\frac{X_{T_{n+1}}}{X_{T_n}}\bigg)^{\alpha}h(Y_{T_{n+1}})+{{\frac{1}{\alpha}}A^*(Y_{T_{n+1}})}\bigg(\frac{X_{T_{n+1}}}{X_{T_n}}\bigg)^{\alpha(1-\gamma)}\bigg|\mathcal{F}_{T_n}\right]
\nonumber\\
\hspace{-0.3cm}&\leq&\hspace{-0.3cm}
e^{-\delta\tau}\sup_{X\in\Tilde{\mathcal{U}}_{T_n,T_{n+1}}(1,Y_{T_n})}\mathbb{E}\left[\frac{1}{\alpha}X^{\alpha}_{T_{n+1}}h(Y_{T_{n+1}})+{\frac{1}{\alpha}}A^*(Y_{T_{n+1}})X^{\alpha(1-\gamma)}\Big|\mathcal{F}_{T_n}\right]
\nonumber\\
\hspace{-0.3cm}&=&\hspace{-0.3cm}
{\frac{1}{\alpha}}\Psi(A^*;Y_{T_n}).
\end{eqnarray}
Recall that $A^*(\cdot)$ is the fixed-point of the function $\Psi(A;\cdot)$. This, together with \eqref{3.4} and \eqref{3.5}, implies
\begin{eqnarray}\label{3.6}
\mathbb{E}[D_{n+1}|\mathcal{F}_{T_n}]\leq D_n+{\frac{1}{\alpha}}e^{-\delta T_n}X_{T_n}^{\alpha(1-\gamma)}\left[\Psi(A^*;Y_{T_n})-A^*(Y_{T_n})\right]=D_n.
\end{eqnarray}
Hence, $(D_n)_{n=0,1,2,...}$ is a $\{\mathcal{F}_{T_n}\}$-supermartingale, and it holds that
\begin{eqnarray}
{\frac{1}{\alpha}}A^*(y)x^{\alpha(1-\gamma)}=D_0\geq\mathbb{E}\left[\sum_{i=1}^ne^{-\delta T_i}\frac{1}{\alpha}\bigg(\frac{X_{T_i}}{X_{T_{i-1}}^{\gamma}}\bigg)^{\alpha}h(Y_{T_{i}})+{\frac{1}{\alpha}}A^*(Y_{T_n})e^{-\delta T_n}X_{T_n}^{\alpha(1-\gamma)}\right],\nonumber
\end{eqnarray}
which yields that
\begin{eqnarray}
\mathbb{E}\left[\sum_{i=1}^ne^{-\delta T_i}\frac{1}{\alpha}\bigg(\frac{X_{T_i}}{X_{T_{i-1}}^{\gamma}}\bigg)^{\alpha}h(Y_{T_{i}})\right]
\hspace{-0.3cm}&\leq&\hspace{-0.3cm}
{\frac{1}{\alpha}}A^*(y)x^{\alpha(1-\gamma)}-{\frac{1}{\alpha}}e^{-\delta T_n}
\mathbb{E}\left[A^*(Y_{T_n})X^{\alpha(1-\gamma)}_{T_n}\right].\nonumber
\end{eqnarray}
By Assumption \ref{ass1} (i.e., $\delta>\zeta(\alpha(1-\gamma))\vee0$), \eqref{3.13.v0}, the fact $A^*(\cdot)\in C^+_b(\mathbb{R})$, and the monotone convergence theorem, we have that
\begin{eqnarray}
\mathbb{E}\left[\sum_{i=1}^{\infty}e^{-\delta T_i}\frac{1}{\alpha}\bigg(\frac{X_{T_i}}{X_{T_{i-1}}^{\gamma}}\bigg)^{\alpha}h(Y_{T_{i}})\right]\leq {\frac{1}{\alpha}}A^*(y)x^{\alpha(1-\gamma)},\nonumber
\end{eqnarray}
and then
\begin{eqnarray}
V(x,y)=\sup_{X\in\mathcal{U}_0(x,y)}\mathbb{E}\left[\sum_{i=1}^{\infty}e^{-\delta T_i}\frac{1}{\alpha}\bigg(\frac{X_{T_i}}{X_{T_{i-1}}^{\gamma}}\bigg)^{\alpha}h(Y_{T_{i}})\right]\leq {\frac{1}{\alpha}}A^*(y)x^{\alpha(1-\gamma)}.\nonumber
\end{eqnarray}
To show the reverse inequality, it suffices to verify the existence of some admissible process $X^*$ such that $$\mathbb{E}\left[\sum_{i=1}^{\infty}e^{-\delta T_i}\frac{1}{\alpha}\left(\frac{X^*_{T_i}}{(X^*_{T_{i-1}})^{\gamma}}\right)^{\alpha}h(Y_{T_i})\right]={\frac{1}{\alpha}}A^*(y)x^{\alpha(1-\gamma)}.$$
By Proposition \ref{prop2.1}, under the choice of $\frac{X_{T_{n+1}}}{X_{T_n}}=x^*_{h_{A^*}}\left(\lambda^*_n\frac{Z^{\eta_{{n+1}}^*}_{T_{n+1}}/B_{T_{n+1}}}{Z^{\eta_{n}^*}_{T_n}/B_{T_n}},Y_{T_{n+1}}\right)$, we have that
\begin{eqnarray}
\hspace{-0.3cm}&&\hspace{-0.3cm}
e^{-\delta\tau}\mathbb{E}\left[\left.\frac{1}{\alpha}\bigg(\frac{X_{T_{n+1}}}{X_{T_n}}\bigg)^{\alpha}h(Y_{T_{n+1}})+{\frac{1}{\alpha}}A^*(Y_{T_{n+1}})\bigg(\frac{X_{T_{n+1}}}{X_{T_n}}\bigg)^{\alpha(1-\gamma)}\right|\mathcal{F}_{T_n}\right]
\nonumber\\
\hspace{-0.3cm}&=&\hspace{-0.3cm}
e^{-\delta\tau}\sup_{X\in\Tilde{\mathcal{U}}_{T_n,T_{n+1}}(1,Y_{T_n})}\mathbb{E}\left[\frac{1}{\alpha}X^{\alpha}_{T_{n+1}}h(Y_{T_{n+1}})+{\frac{1}{\alpha}}A^*(Y_{T_{n+1}})X^{\alpha(1-\gamma)}\Big|\mathcal{F}_{T_n}\right]={\frac{1}{\alpha}}\Psi(A^*;Y_{T_n}),\nonumber
\end{eqnarray}
and, for all $n$,
\begin{eqnarray}\label{3.11.power}
\hspace{-0.3cm}&&\hspace{-0.3cm}
\mathbb{E}\left[\frac{Z^{\eta^*_{n+1}}_{T_{n+1}}/B_{T_{n+1}}}{Z^{\eta^*_{n}}_{T_n}/B_{T_n}}x^*_{h_{A^*}}\left(\lambda^*_n\frac{Z^{\eta^*_{n+1}}_{T_{n+1}}/B_{T_{n+1}}}{Z^{\eta^*_n}_{T_n}/B_{T_n}},Y_{T_{n+1}}\right)\Bigg|\mathcal{F}_{T_n}\right]
\nonumber\\
\hspace{-0.3cm}&=&\hspace{-0.3cm}
\mathbb{E}\left[\frac{Z^{\eta^*_1}_{T_{1}}/B_{T_1}}{Z^{\eta^*_0}_{T_0}/B_{T_0}}x^*_{h_{A^*}}\left(\lambda^*_n\frac{Z^{\eta^*_1}_{T_{1}}/B_{T_1}}{Z^{\eta^*_0}_{T_0}/B_{T_0}},Y_{T_1}\right)\bigg|\mathcal{F}_{T_0}\right]=\mathbb{E}\left[x^*_{h_{A^*}}\left(\lambda^*\frac{Z^{\eta^*}_{\tau}}{B_{\tau}},Y_{\tau}\right)\frac{Z^{\eta^*}_{\tau}}{B_{\tau}}\right]=1.
\end{eqnarray}
Next, define a sequence of random variables recursively that
\begin{eqnarray}
Q_n:=\left\{
\begin{array}{ll}
    xx^*_{h_{A^*}}\left(\lambda^*\frac{Z^{\eta_1^*}_{T_1}}{B_{T_1}},Y_{T_1}\right), & n=1, \\
     Q_{n-1}x^*_{h_{A^*}}\left(\lambda^*\frac{Z^{\eta^*_{n}}_{T_{n}}/B_{T_{n}}}{Z^{\eta^*_{n-1}}_{T_{n-1}}/B_{T_{n-1}}},Y_{T_n}\right),& n=2,3,....
\end{array}
\right.\nonumber
\end{eqnarray}
By \eqref{3.11.power}, it holds that
\begin{eqnarray}
\mathbb{E}\left[\frac{Z^{\eta^*_{n+1}}_{T_{n+1}}}{B_{T_{n+1}}}Q_{{n+1}}\bigg|\mathcal{F}_{T_n}\right]=\mathbb{E}\left[\frac{Z^{\eta^*_{n+1}}_{T_{n+1}}}{B_{T_n}}Q_{n}x^*_{h_{A^*}}\left(\lambda^*\frac{Z^{\eta^*_{n+1}}_{T_{n+1}}/B_{T_{n+1}}}{Z^{\eta^*_{n}}_{T_{n}}/B_{T_n}},Y_{T_{n+1}}\right)\Bigg|\mathcal{F}_{T_n}\right]=\frac{Z^{\eta^*_{n}}_{T_n}}{B_{T_n}}Q_{n},\nonumber
\end{eqnarray}
which yields that $(\frac{Z^{\eta^*_{n}}_{T_n}}{B_{T_n}}Q_{n})_{n\geq 1}$ is a $\{\mathcal{F}_{T_n}\}$-martingale. Additionally, by Lemma \ref{prop.euquiv},
there exists some $X^*\in\mathcal{U}_0(x,y)$ such that $X^*_{T_n}=Q_n$ for all $n$. Define
$$D^*_n:=\sum_{i=1}^ne^{-\delta T_i}\frac{1}{\alpha}\bigg(\frac{X^*_{T_i}}{{(X^*_{T_{i-1}})}^{\gamma}}\bigg)^{\alpha}h(Y_{T_{i}})+{\frac{1}{\alpha}}A^*(Y_{T_n})e^{-\delta T_n}{(X^*_{T_n})}^{\alpha(1-\gamma)}, \quad n\geq 0.$$
Using the same arguments leading to \eqref{3.6}, we are ready to conclude that
\begin{eqnarray}
\mathbb{E}[D^*_{n+1}|\mathcal{F}_{T_n}]=D^*_n+{\frac{1}{\alpha}}e^{-\delta T_n}{(X^*_{T_n})}^{\alpha(1-\gamma)}\left[\Psi(A^*(Y_{T_n}))-A^*(Y_{T_n})\right]=D^*_n, \quad n\geq 0,\nonumber
\end{eqnarray}
which gives that $D^*_n$ is a $\{\mathcal{F}_{T_n}\}$-martingale. Hence, we have
\begin{eqnarray}
\mathbb{E}\left[\sum_{i=1}^{n}e^{-\delta T_i}\frac{1}{\alpha}\bigg(\frac{X^*_{T_i}}{(X^*_{T_{i-1}})^{\gamma}}\bigg)^{\alpha}h(Y_{T_{i}})\right]
\hspace{-0.3cm}&=&\hspace{-0.3cm}
{\frac{1}{\alpha}}A^*(y)x^{\alpha(1-\gamma)}-{\frac{1}{\alpha}}e^{-\delta T_n}\mathbb{E}\left[A^*(Y_{T_n})(X^*_{T_n})^{\alpha(1-\gamma)}\right].\nonumber
\end{eqnarray}
By Assumption \ref{ass1} and the monotone convergence theorem, we have
\begin{eqnarray}
\mathbb{E}\left[\sum_{i=1}^{\infty}e^{-\delta T_i}\frac{1}{\alpha}\bigg(\frac{X^*_{T_i}}{(X^*_{T_{i-1}})^{\gamma}}\bigg)^{\alpha}h(Y_{T_{i}})\right]={\frac{1}{\alpha}}A^*(y)x^{\alpha(1-\gamma)}.\nonumber
\end{eqnarray}
Finally, it follows from Lemma \ref{prop.euquiv} that the optimal wealth process admits the stochastic integral representation
$$\frac{X^*_t}{B_t}=X^*_{T_n}
+\int_{T_n}^t\psi^*_{nt}dW^{\eta_{n}^*}_{1t},\quad T_n\leq t<T_{n+1},$$
for some ${\mathcal{F}_t}$-adapted process $\psi^*_n$ satisfying $\int_{T_n}^{T_{n+1}}{\psi^*_{nt}}^2dt<\infty$ almost sure,
which, combined with \eqref{dX/B}, then implies \eqref{pi^*}.
The proof is complete.
\end{proof}

\section{Periodic Evaluation under Logarithmic Utility}\label{sec:log}

In this section, we consider the case of the logarithmic utility that
\begin{eqnarray}\label{log}
U(x,y):=\log x+h(y),\quad (x,y)\in\mathbb{R}_+\times\mathbb{R},
\end{eqnarray}
where $\mathbb{R}\ni y\mapsto h(y)\in\mathbb{R}_+$ is a bounded continuous function such that, $m\leq h(y)\leq 1$ for any $y\in\mathbb{R}$ with $m\in(0,1)$. 

The value function of the
portfolio optimization problem under periodic evaluation is defined by
\begin{eqnarray}\label{problem.log}
V(x,y)
\hspace{-0.3cm}&:=&\hspace{-0.3cm}
\sup_{X\in\mathcal{U}_0(x,y)}\mathbb{E}\left[\sum_{i=1}^{\infty}e^{-\delta T_i}\left(\log\frac{X_{T_i}}{\left(X_{T_{i-1}}\right)^{\gamma}}+h(Y_{T_i})\right)\right],
\end{eqnarray}
where $\mathcal{U}_0$ is the set of admissible portfolio processes defined by \eqref{admissible_set}.
Again, using the dynamic programming principle, we can heuristically derive that
\begin{eqnarray}\label{ddp.dual}
V(x,y)
\hspace{-0.3cm}&=&\hspace{-0.3cm}
\sup_{X\in\mathcal{U}_0(x,y)}\mathbb{E}\left[e^{-\delta \tau}\left(\log\frac{X_{\tau}}{x^{\gamma}}+h(Y_{\tau})\right)+e^{-\delta \tau}V(X_{\tau},Y_{\tau})\right],\quad (x,y)\in\mathbb{R}_+\times\mathbb{R}.
\end{eqnarray}

Thanks to the structure of the logarithmic utility and the fact that it is bounded with respect to $y$, we heuristically conjecture that our value function would have the form of $V(x,y)=A^*(y)+C^*\log x$ for some continuous, bounded, and non-negative function $A^*(\cdot)\in C^+_b(\mathbb{R})$ as well as some constant $C^*\in\mathbb{R}$. Substituting this form of $V$ into \eqref{ddp.dual} and then subtract the both sides by $C^*\log x$, one obtains
\begin{eqnarray}\label{A*}
A^*(y)
\hspace{-0.3cm}&=&\hspace{-0.3cm}
\sup_{X\in\mathcal{U}_0(x,y)}\mathbb{E}\left[e^{-\delta T_1}\left(\log\frac{X_{T_1}}{x^{\gamma}}+h(Y_{T_1})\right)+e^{-\delta T_1}[A^*(Y_{T_1})+C^*\log X_{T_1}]-C^* \log x\right]
\nonumber\\
\hspace{-0.3cm}&=&\hspace{-0.3cm}
\sup_{X\in\mathcal{U}_0(x,y)}\mathbb{E}\left[e^{-\delta \tau}(1+C^*)\log \frac{X_{\tau}}{x}+[C^*(e^{-\delta\tau}-1)+(1-\gamma) e^{-\delta \tau}]\log x\right.
\nonumber\\
\hspace{-0.3cm}&&\hspace{2cm}
\left.+e^{-\delta \tau}\left(h(Y_{\tau})+A^*(Y_{\tau})\right)\right].
\end{eqnarray}
By setting $C^*=\frac{1-\gamma}{e^{\delta\tau}-1}$ in the above equation, \eqref{A*} can be reformulated as
\begin{eqnarray}\label{A*.dual}
A^*(y)=\frac{1-\gamma e^{-\delta\tau}}{e^{\delta\tau}-1}\sup_{X\in\mathcal{U}_0(1,y)}\mathbb{E}\left[\log X_{\tau}\right]+e^{-\delta \tau}\mathbb{E}\left[h(Y_{\tau})+A^*(Y_{\tau})\right].
\end{eqnarray}
Hence, in the case of logarithmic utility, the characterization of $V$ now simplifies to the characterization of the unknown, non-negative, continuous, and bounded function $A^*(\cdot)$. We again first solve the one-period terminal wealth utility maximization problem:
\begin{eqnarray}\label{log.pri}
\sup_{X\in\mathcal{U}_0(x,y)}\mathbb{E}\left[\log X_{\tau}\right],\quad (x,y)\in\mathbb{R}_+\times\mathbb{R}.
\end{eqnarray}
To this end, we state the original problem \eqref{log.pri} equivalently as
\begin{eqnarray}\label{log.pri2}
\sup_{X\in\Tilde{\mathcal{U}}_{0,\tau}(x,y)}\mathbb{E}[\log X_{\tau}],
\end{eqnarray}
where
\begin{eqnarray}\label{tilde.U}
\Tilde{\mathcal{U}}_{s,t}(x,y):=\left\{X\in\mathcal{F}^+_{t}:\sup_{\eta\in\mathcal{M}}\mathbb{E}\left[\frac{Z^{\eta}_{t}/B_t}{Z^{\eta}_s/B_{s}}X\bigg|\mathcal{F}_s\right]\leq x\text{ and }\mathbb{E}\left[\left(\log{X}_t+h(Y_{t})\right)_-|\mathcal{F}_s\right]<\infty\right\},
\end{eqnarray}
with $0\leq s\leq t<\infty$, $X_s=x$ and $Y_s=y$. 
Indeed, it is straightforward that when $X\in\mathcal{U}_{0}(x,y)$ and $\eta\in\mathcal{M}$, according to \eqref{supermaringale}, the process $\frac{XZ^{\eta}}{B}$ is a $\mathbb{P}$-supermartingale. This, together with the arbitrariness of $\eta$, yields that
\begin{eqnarray}\label{sup<x}
\sup_{\eta\in\mathcal{M}}\mathbb{E}\left[\frac{X_{\tau}Z^{\eta}_{\tau}}{B_{\tau}}\right]\leq x.
\end{eqnarray}
In addition, note that $$e^{-\delta\tau}\mathbb{E}\left[\left(\log\frac{X_{\tau}}{x^{\gamma}}+h(Y_{\tau})\right)_-\right]\leq\sum_{i=1}^{\infty}e^{-\delta T_i}\mathbb{E}\left[\left(\log\frac{X_{T_i}}{(X_{T_{i-1}})^{\gamma}}+h(Y_{T_i})\right)_-\right]<\infty,$$ which implies that $\mathbb{E}[(\log X_{\tau}+h(Y_{\tau}))_-]<\infty$. This, together with \eqref{sup<x}, yields that, $X_{\tau}\in\Tilde{\mathcal{U}}_{0,\tau}(x,y)$ for any $X\in\mathcal{U}_{0}(x,y)$. Therefore, one gets $$\sup_{X\in\mathcal{U}_{0}(x,y)}\mathbb{E}\left[\log X_{\tau}\right]\leq \sup_{X\in\Tilde{\mathcal{U}}_{0,\tau}(x,y)}\mathbb{E}\left[\log X_{\tau}\right].$$ We claim that the reverse inequality also holds. In fact, for any $X\in\Tilde{\mathcal{U}}_{0,\tau}(x,y)$, one can follow the same arguments of Lemma 2.2 in \cite{CH05} to find a portfolio proportion $\pi$ and the wealth process $X^{x,y,\pi}\in\mathcal{U}_{0}(x,y)$ (as we can define $(\pi_t)_{t\in[0,\tau]}$ in the same manner as Lemma 2.2 in \cite{CH05} and let $\pi_t\equiv0$ for $t\in(\tau,\infty)$) such that $X^{x,y,\pi}_{\tau}\geq X$, and hence, the reverse inequality also holds true.

As in previous section, we shall show that the optimal solution to optimization problem \eqref{log.pri2} exists and there indeed exists a unique non-negative $A^*(\cdot)\in {C}_b^{+}(\mathbb{R})$ solving equation \eqref{A*} and then to formally prove that $A^*(y)+C^*\log x$ coincides with the value function $V(\cdot,\cdot)$ via the verification theorem.

Following \cite{CH05}, for fixed $(x,y)\in\mathbb{R}_+\times\mathbb{R}$, we consider the dual functional as
\begin{eqnarray}\label{L.log}
L(\eta,\lambda):=L(\eta,\lambda;x,y)
\hspace{-0.3cm}&:=&\hspace{-0.3cm}
\sup_{X\in\mathcal{F}_{\tau}^+}\left\{\mathbb{E}\left[\log X-\lambda\frac{XZ_{\tau}^{\eta}}{B_{\tau}}\right]\right\}+\lambda x, 
\quad (\eta,\lambda)\in\mathcal{M}\times\mathbb{R}_+.
\end{eqnarray}
The dual problem associated to the primal problem \eqref{log.pri} is defined by
\begin{eqnarray}\label{dual.log}
\text{minimizing}\quad L(\eta,\lambda),\quad\text{over}\quad(\eta,\lambda)\in\mathcal{M}\times\mathbb{R}_+.
\end{eqnarray}



By Proposition 3.1 and the first result of Remark 4.6 in \cite{CH05}, we have the following duality characterization of the optimal solution $X^*$ to the primal problem \eqref{log.pri2}.
\begin{prop}\label{prop2.2} 
For fixed $(x,y)\in\mathbb{R}_+\times\mathbb{R}$, put $\eta^*\equiv0$, $\lambda^*=\frac{1}{x}$ and
$X^*:=\frac{B_{\tau}}{\lambda^*Z^{\eta^*}_{\tau}}$.
It holds that $X^*$ is the optimal solution to the primal problem \eqref{log.pri2}, and $(\eta^*,\lambda^*)$ is the optimal solution to the dual problem \eqref{dual.log}. In particular, it holds that
\begin{eqnarray}\label{sup<inf}
\inf_{\eta\in\mathcal{M},\lambda>0}L(\eta,\lambda)
\hspace{-0.3cm}&=&\hspace{-0.3cm}
\sup_{X\in\Tilde{\mathcal{U}}_{0,\tau}(x,y)}\mathbb{E}\left[\log X_{\tau}\right],\quad (x,y)\in\mathbb{R}_+\times\mathbb{R}.
\end{eqnarray}
\end{prop}

Thanks to the characterization of the optimal solution to primal problem \eqref{log.pri2} in Proposition \ref{prop2.2}, we can obtain the upper and lower bounds for the value function $V(\cdot,\cdot)$ given by \eqref{problem.log}.

\begin{prop}
For fixed $(x,y)\in\mathbb{R}_+\times\mathbb{R}$, it holds that
\begin{eqnarray}\label{V.log}
\frac{1-\gamma}{e^{\delta\tau}-1}\log x+\frac{m}{e^{\delta\tau}-1}+r\tau\frac{e^{\delta\tau}-\gamma}{(e^{\delta\tau}-1)^2}
\hspace{-0.3cm}&\leq&\hspace{-0.3cm}
V(x,y)
\nonumber\\
\hspace{-0.3cm}&\leq&\hspace{-0.3cm}
\frac{1-\gamma}{e^{\delta\tau}-1}\log x+\frac{1}{e^{\delta\tau}-1}+\tau(r+\frac{M_0}{2})\frac{e^{\delta\tau}-\gamma}{(e^{\delta\tau}-1)^2},
\nonumber
\end{eqnarray}
where $m\in(0,1)$ is the lower bound of the function $h(\cdot)$ that appears in the utility function $U(\cdot,\cdot)$ given by \eqref{log} and $M_0$ is a positive constant given by \eqref{M.def}.
\end{prop}
\begin{proof}
The lower bound of $V(x,y)$ can be derived by noting that $(xe^{rt})_{t\geq 0}$ is an admissible portfolio process in $\mathcal{U}_0(x,y)$. More precisely, using the lower bound of $h(\cdot)$, one has
\begin{eqnarray}
\mathbb{E}\left[\sum_{i=1}^{\infty}e^{-\delta T_i}\bigg(\log\frac{X_{T_i}}{X_{T_{i-1}}^{\gamma}}+h(Y_{T_i})\bigg)\right]
\hspace{-0.3cm}&=&\hspace{-0.3cm}
\mathbb{E}\left[\sum_{i=1}^{\infty}e^{-\delta T_{i}}\left(\log\frac{xe^{rT_i}}{(xe^{rT_{i-1}})^{\gamma}}+h(Y_{T_{i}})\right)\right]
\nonumber\\
\hspace{-0.3cm}&\geq&\hspace{-0.3cm} 
\sum_{i=1}^{\infty}e^{-\delta T_{i}}\left[m+(1-\gamma)\log x+r\tau(i(1-\gamma)+\gamma)\right]
\nonumber\\
\hspace{-0.3cm}&=&\hspace{-0.3cm}
\frac{1-\gamma}{e^{\delta\tau}-1}\log x+\frac{m}{e^{\delta\tau}-1}+r\tau\frac{e^{\delta\tau}-\gamma}{(e^{\delta\tau}-1)^2}.
\end{eqnarray}

We then establish the upper bound. For any admissible portfolio process $(X_{t})_{t\geq0}\in\mathcal{U}_0(x,y)$ and $n\geq1$, using the fact of $m\leq h(y)\leq 1$ with $m\in(0,1)$, one has
\begin{eqnarray}\label{cond.1}
\hspace{-0.3cm}&&\hspace{-0.3cm}
\mathbb{E}\left[\sum_{i=1}^{n}e^{-\delta T_i}\bigg(\log\frac{X_{T_i}}{X_{T_{i-1}}^{\gamma}}+h(Y_{T_i})\bigg)\right]=
\sum_{i=1}^{n}e^{-\delta T_{i}}\mathbb{E}\left[\log\frac{X_{T_i}}{X_{T_{i-1}}^{\gamma}}+h(Y_{T_{i}})\right]
\nonumber\\
\hspace{-0.3cm}&\leq&\hspace{-0.3cm}
\sum_{i=1}^{n}e^{-\delta T_{i}}+\sum_{i=1}^{n}e^{-\delta T_{i}}\mathbb{E}\left[\log\frac{X_{T_i}}{X_{T_{i-1}}}+(1-\gamma)\log X_{T_{i-1}}\right]
\nonumber\\
\hspace{-0.3cm}&=&\hspace{-0.3cm}
\sum_{i=1}^{n}e^{-\delta T_{i}}+\sum_{i=1}^{n}e^{-\delta T_{i}}\mathbb{E}\left[\mathbb{E}\left[\left.\log\frac{X_{T_i}}{X_{T_{i-1}}}\right|\mathcal{F}_{T_{i-1}}\right]+(1-\gamma)\log X_{T_{i-1}}\right]
\nonumber\\
\hspace{-0.3cm}&\leq&\hspace{-0.3cm}
\sum_{i=1}^{n}e^{-\delta T_{i}}+\sum_{i=1}^{n}e^{-\delta T_{i}}
\mathbb{E}\left[\left(\sup_{X\in\mathcal{U}_{T_{i-1},T_i}(1,Y_{T_{i-1}})}\mathbb{E}\left[\log{X}_{T_i}\big|\mathcal{F}_{T_{i-1}}\right]\right)+(1-\gamma)\log X_{T_{i-1}}\right].
\end{eqnarray}
Note that $\sup_{X\in\mathcal{U}_{T_{i-1},T_i}(1,Y_{T_{i-1}})}\mathbb{E}[\log X_{T_i}|\mathcal{F}_{T_{i-1}}]$ is the value function of the investment problem with maturity $\tau$.
Moreover, by  Proposition \ref{prop2.2}, it is possible to derive an upper bound for the primal problem $\sup_{X\in\mathcal{U}_{T_{i-1},T_i}(1,Y_{T_{i-1}})}\mathbb{E}[\log X_{T_i}|\mathcal{F}_{T_{i-1}}]$ by straightforward calculus. Specifically, using the results in Proposition \ref{prop2.2}, one gets
\begin{eqnarray}
\label{upper}
\sup_{X\in\Tilde{\mathcal{U}}_{T_{i-1},T_i}(1,Y_{T_{i-1}})}\mathbb{E}[\log X_{T_i}|\mathcal{F}_{T_{i-1}}]
\hspace{-0.3cm}&=&\hspace{-0.3cm}
r\tau-\mathbb{E}[\log Z^0_{T_i}|\mathcal{F}_{T_{i-1}}]
\nonumber\\
\hspace{-0.3cm}&=&\hspace{-0.3cm}
r\tau+\mathbb{E}\left[\int_{T_{i-1}}^{T_i}\theta(Y_s)dW_{1s}+\frac{1}{2}\int_{T_{i-1}}^{T_i}\theta^2(Y_s)ds\bigg|\mathcal{F}_{T_{i-1}}\right]
\nonumber\\
\hspace{-0.3cm}&\leq&\hspace{-0.3cm}
(r+\frac{M_0}{2})\tau.
\end{eqnarray}
Similarly, it holds that
\begin{eqnarray}\label{cond.2}
\sup_{X\in\Tilde{\mathcal{U}}_{0,T_{i-1}}(x,y)}\mathbb{E}\left[\log{X}_{T_{i-1}}\right]\leq \left(r+\frac{M_0}{2}\right)(i-1)\tau+\log x.
\end{eqnarray}
Combining \eqref{cond.1}-\eqref{cond.2}, one obtains
\begin{eqnarray}
\hspace{-0.3cm}&&\hspace{-0.3cm}
\mathbb{E}\left[\sum_{i=1}^{\infty}e^{-\delta T_i}\bigg(\log\frac{X_{T_i}}{X_{T_{i-1}}^{\gamma}}+h(Y_{T_i})\bigg)\right]
\nonumber\\
\hspace{-0.3cm}&\leq&\hspace{-0.3cm}
\sum_{i=1}^{\infty}e^{-\delta T_{i}}(1+(r+\frac{M_0}{2})\tau)+(1-\gamma)\sum_{i=1}^{\infty}e^{-\delta T_{i}}\mathbb{E}\left[\log X_{T_{i-1}}\right]
\nonumber\\
\hspace{-0.3cm}&\leq&\hspace{-0.3cm}
\sum_{i=1}^{\infty}e^{-\delta T_{i}}\left[1+(r+\frac{M_0}{2})\tau+(1-\gamma)(r+\frac{M_0}{2})(i-1)\tau+(1-\gamma)\log x\right],\nonumber
\end{eqnarray}
from which the desired upper bound in \eqref{V.log} follows.
\end{proof}

Similar to the proof for power utility, we can also verify the existence of fixed point $A^*(\cdot)\in C^+_b(\mathbb{R})$ solving equation \eqref{A*.dual}. Hence, we present the following Proposition \ref{fixed.log} without proof.

\begin{prop}\label{fixed.log}
Define a functional $\Psi: {C}_b^{+}(\mathbb{R})\mapsto {C}_b^{+}(\mathbb{R})$ as
\begin{eqnarray}\label{Psi.2}
\Psi(A):=\Psi(A;y):=\frac{1-\gamma e^{-\delta\tau}}{e^{\delta\tau}-1}\sup_{X\in\Tilde{\mathcal{U}}_{0,\tau}(1,y)}\mathbb{E}\left[\log X_{\tau}\right]+e^{-\delta \tau}\mathbb{E}\left[h(Y_{\tau})+A(Y_{\tau})\right],\quad y\in\mathbb{R}.
\end{eqnarray} 
Then $\Psi$ is a contraction on the metric space $({C}_b^{+}(\mathbb{R}),d)$, where the metric $d$ is defined by $d(x,y):=\sup_{t\in\mathbb{R}}|x(t)-y(t)|$, $x,y\in {C}_b^{+}(\mathbb{R})$. Furthermore, $\Psi$ admits a unique fixed-point $A^*(\cdot)$ such that $A^*(y)=\Psi(A^*;y)$ for all $y\in\mathbb{R}$,
and $A^*(\cdot)$ satisfies
    $$\frac{e^{\delta\tau}-\gamma}{(e^{\delta\tau}-1)^2}r\tau+\frac{me^{-\delta\tau}}{1-e^{-\delta\tau}}\leq A^*(y)\leq \frac{e^{\delta\tau}-\gamma}{(e^{\delta\tau}-1)^2}(r+\frac{M_0}{2})\tau+\frac{e^{-\delta\tau}}{1-e^{-\delta\tau}},\quad y\in\mathbb{R}.$$
\end{prop}

With the previous preparations, we are ready to present the main result in this section.
\begin{thm}
The value function of the problem \eqref{problem.log} is given by
\begin{eqnarray}
V(x,y)=A^*(y)+C^*\log x, \quad (x,y)\in\mathbb{R}_+\times\mathbb{R},\nonumber
\end{eqnarray}
where $A^*(\cdot)$ is the unique fixed-point of the function $\Psi$ defined in \eqref{Psi.2} and $C^*=\frac{1-\gamma}{e^{\delta\tau}-1}$.
In addition, the optimal wealth process is given by
\begin{eqnarray}
X^*_{t}:=x\frac{B_{t}}{Z_{t}^0},\quad t\in[0,\infty),\nonumber
\end{eqnarray}
where the process $Z_t^0$ is given by \eqref{Z.process} with $\eta\equiv0$. The optimal proportion of wealth invested in the risky asset is given by
\begin{eqnarray}
\pi^*_t=\frac{\mu(Y_t)-r}{\sigma^2(Y_t)},\quad t\in [0,\infty).\nonumber
\end{eqnarray}
\end{thm}
\begin{proof}
For any admissible wealth process $(X_{t})_{t\geq 0}\in\mathcal{U}_0(x,y)$, define a discrete-time stochastic process $D=(D_n)_{n\geq 0}$ by
$$D_n:=\sum_{i=1}^ne^{-\delta T_i}\bigg(\log\frac{X_{T_i}}{X_{T_{i-1}}^{\gamma}}+h(Y_{T_i})\bigg)+e^{-\delta T_n}(A^*(Y_{T_n})+C^*\log X_{T_n}),\quad n\geq 1,$$
and $D_0=A^*(y)+C^*\log x$.
Then
\begin{eqnarray}\label{Dn+1.log}
D_{n+1}
\hspace{-0.3cm}&=&\hspace{-0.3cm}
D_{n}+e^{-\delta T_{n+1}}\bigg(\log\frac{X_{T_{n+1}}}{X_{T_{n}}^{\gamma}}+h(Y_{T_{n+1}})\bigg)-e^{-\delta T_n}(A^*(Y_{T_n})+C^*\log X_{T_n})
\nonumber\\
\hspace{-0.3cm}&&\hspace{-0.3cm}
+e^{-\delta T_{n+1}}(A^*(Y_{T_{n+1}})+C^*\log X_{T_{n+1}})
\nonumber\\
\hspace{-0.3cm}&=&\hspace{-0.3cm}
D_n+e^{-\delta T_n}\left[e^{-\delta\tau}\left(\log\frac{X_{T_{n+1}}}{X^{\gamma}_{T_{n}}}+h(Y_{T_{n+1}})+(A^*(Y_{T_{n+1}})+C^*\log X_{T_{n+1}})\right)\right.
\nonumber\\
\hspace{-0.3cm}&&\hspace{-0.3cm}
\left.-(A^*(Y_{T_n})+C^*\log X_{T_n})\right].
\end{eqnarray}
Taking the conditional expectation on both sides of \eqref{Dn+1.log}, one gets
\begin{eqnarray}\label{3.4.log}
\mathbb{E}[D_{n+1}|\mathcal{F}_{T_n}]
\hspace{-0.3cm}&=&\hspace{-0.3cm}
D_n+e^{-\delta T_n}\left[\mathbb{E}\left[\frac{1-\gamma e^{-\delta\tau}}{e^{\delta\tau}-1}\log\frac{X_{T_{n+1}}}{X_{T_n}}\right.\right.
\left.\left.+e^{-\delta\tau}(h(Y_{T_{n+1}})+A^*(Y_{T_{n+1}}))\bigg|\mathcal{F}_{T_n}\right]-A^*(Y_{T_n})\right].
\nonumber\\
\end{eqnarray}
Furthermore, according to the definition of $\Psi$ (see, \eqref{Psi.2}), one knows
\begin{eqnarray}\label{3.5.log}
\hspace{-0.3cm}&&\hspace{-0.3cm}
\mathbb{E}\left[\frac{1-\gamma e^{-\delta\tau}}{e^{\delta\tau}-1}\log\frac{X_{T_{n+1}}}{X_{T_n}}+e^{-\delta\tau}(h(Y_{T_{n+1}})+A^*(Y_{T_{n+1}}))\bigg|\mathcal{F}_{T_n}\right]
\nonumber\\
\hspace{-0.3cm}&\leq&\hspace{-0.3cm}
\sup_{X\in\Tilde{\mathcal{U}}_{T_n,T_{n+1}}(1,Y_{T_n})}\mathbb{E}\left[\frac{1-\gamma e^{-\delta\tau}}{e^{\delta\tau}-1}\log X_{T_{n+1}}+e^{-\delta\tau}(h(Y_{T_{n+1}})+A^*(Y_{T_{n+1}}))\bigg|\mathcal{F}_{T_n}\right]
\nonumber\\
\hspace{-0.3cm}&=&\hspace{-0.3cm}
\Psi(A^*;Y_{T_n}).
\end{eqnarray}
Recall that $A^*(\cdot)$ is the fixed point of the function $\Psi(A;\cdot)=A(\cdot)$. This, together with \eqref{3.4.log} and \eqref{3.5.log}, implies
\begin{eqnarray}\label{3.6.log}
\mathbb{E}[D_{n+1}|\mathcal{F}_{T_n}]\leq D_n+e^{-\delta T_n}\left[\Psi(A^*;Y_{T_n})-A^*(Y_{T_n})\right]=D_n.
\end{eqnarray}
Hence, $(D_n)_{n\geq 0}$ is a $\{\mathcal{F}_{T_n}\}$-supermartingale, and then, one gets
\begin{eqnarray}\label{A+C+=D}
\hspace{-0.3cm}&&\hspace{-0.3cm}
A^*(y)+C^*\log x=D_0
\nonumber\\
\hspace{-0.3cm}&\geq&\hspace{-0.3cm}
\mathbb{E}\left[\sum_{i=1}^ne^{-\delta T_i}\left(\log\frac{X_{T_i}}{X_{T_{i-1}}^{\gamma}}+h(Y_{T_i})\right)+e^{-\delta T_n}(A^*(Y_{T_n})+C^*\log X_{T_n})\right]
\nonumber\\
\hspace{-0.3cm}&=&\hspace{-0.3cm}
\mathbb{E}\left[\sum_{i=1}^ne^{-\delta T_i}\left(\log\frac{X_{T_i}}{X_{T_{i-1}}^{\gamma}}\right)_+-\sum_{i=1}^ne^{-\delta T_i}\left(\log\frac{X_{T_i}}{X_{T_{i-1}}^{\gamma}}\right)_-\right.
\nonumber\\
\hspace{-0.3cm}&&\hspace{0.3cm}
\left.+\sum_{i=1}^ne^{-\delta T_i}h(Y_{T_i})+e^{-\delta T_n}(A^*(Y_{T_n})+C^*\log X_{T_n})\right].
\end{eqnarray}
By the definition of the admissible set, it holds that
\begin{eqnarray}\label{standing ass.}
\sum_{i=1}^{\infty}e^{-\delta T_i}\mathbb{E}\left[\left(\log\frac{X_{T_i}}{\left(X_{T_{i-1}}\right)^{\gamma}}\right)_{-}\right]<\infty,
\end{eqnarray}
and we obtain that
\begin{eqnarray}
\mathbb{E}\left[\sum_{i=1}^ne^{-\delta T_i}\left(\log\frac{X_{T_i}}{X_{T_{i-1}}^{\gamma}}\right)_-\right]
\hspace{-0.3cm}&\leq&\hspace{-0.3cm}
\sum_{i=1}^{\infty}e^{-\delta T_i}\mathbb{E}\left[\left(\log\frac{X_{T_i}}{X_{T_{i-1}}^{\gamma}}\right)_-\right]<\infty,\quad n\geq1.\nonumber
\end{eqnarray}
Hence, one can rewrite \eqref{A+C+=D} as
\begin{eqnarray}
A^*(y)+C^*\log x
\hspace{-0.3cm}&\geq&\hspace{-0.3cm}
\mathbb{E}\left[\sum_{i=1}^ne^{-\delta T_i}\left(\log\frac{X_{T_i}}{X_{T_{i-1}}^{\gamma}}\right)_+\right]-\mathbb{E}\left[\sum_{i=1}^ne^{-\delta T_i}\left(\log\frac{X_{T_i}}{X_{T_{i-1}}^{\gamma}}\right)_-\right]
\nonumber\\
\hspace{-0.3cm}&&\hspace{0.3cm}
+\mathbb{E}\left[\sum_{i=1}^ne^{-\delta T_i}h(Y_{T_i})+e^{-\delta T_n}(A^*(Y_{T_n})+C^*\log X_{T_n})\right].\nonumber
\end{eqnarray}
This, together with $m\leq h(\cdot)\leq1$, \eqref{standing ass.}, Proposition \ref{fixed.log} and Lemma 5.3 of \cite{WYY23}, the monotone convergence theorem, as well as the Fubini's theorem,  yields that
\begin{eqnarray}
\mathbb{E}\left[\sum_{i=1}^{\infty}e^{-\delta T_i}\left(\log\frac{X_{T_i}}{X_{T_{i-1}}^{\gamma}}\right)_+\right]
\hspace{-0.3cm}&\leq&\hspace{-0.3cm}
A^*(y)+C^*\log x+\mathbb{E}\left[\sum_{i=1}^{\infty}e^{-\delta T_i}\left(\left(\log\frac{X_{T_i}}{X_{T_{i-1}}^{\gamma}}\right)_-+h(Y_{T_i})\right)\right]
\nonumber\\
\hspace{-0.3cm}&\leq&\hspace{-0.3cm}
A^*(y)+C^*\log x+\sum_{i=1}^{\infty}e^{-\delta T_i}\mathbb{E}\left[\left(\log\frac{X_{T_i}}{X_{T_{i-1}}^{\gamma}}\right)_-+h(Y_{T_i})\right]
\nonumber\\
\hspace{-0.3cm}&<&\hspace{-0.3cm}
\infty.\nonumber
\end{eqnarray}
Therefore, the random variable $\sum_{i=1}^{\infty}e^{-\delta T_i}\left(\log\frac{X_{T_i}}{X_{T_{i-1}}^{\gamma}}\right)_{+}$ is integrable. Recall that, by \eqref{standing ass.} and the Fubini's theorem, the random variable $\sum_{i=1}^{\infty}e^{-\delta T_i}\left(\log\frac{X_{T_i}}{X_{T_{i-1}}^{\gamma}}\right)_{-}$ is also integrable. In addition, one can verify that
\begin{align}\label{Control function}
    \left|\sum_{i=1}^{n}e^{-\delta T_i}\left(\log\frac{X_{T_i}}{X_{T_{i-1}}^{\gamma}}\right)\right|
\leq & \sum_{i=1}^{n}e^{-\delta T_i}\left(\log\frac{X_{T_i}}{X_{T_{i-1}}^{\gamma}}\right)_{+}
    +
    \sum_{i=1}^{n}e^{-\delta T_i}\left(\log\frac{X_{T_i}}{X_{T_{i-1}}^{\gamma}}\right)_{-}
    \nonumber\\
    \leq
&    \sum_{i=1}^{\infty}e^{-\delta T_i}\left(\log\frac{X_{T_i}}{X_{T_{i-1}}^{\gamma}}\right)_{+}
    +
    \sum_{i=1}^{\infty}e^{-\delta T_i}\left(\log\frac{X_{T_i}}{X_{T_{i-1}}^{\gamma}}\right)_{-},\quad n\geq 1,
\end{align}
where the random variable on the right hand side of \eqref{Control function} is integrable that
\begin{align}
\label{control func.integrable}
\mathbb{E}\left[\sum_{i=1}^{\infty}e^{-\delta T_i}\left(\log\frac{X_{T_i}}{X_{T_{i-1}}^{\gamma}}\right)_{+}
    +
    \sum_{i=1}^{\infty}e^{-\delta T_i}\left(\log\frac{X_{T_i}}{X_{T_{i-1}}^{\gamma}}\right)_{-}\right]<\infty.
\end{align}
By \eqref{control func.integrable}, we know that
\begin{align}
    \sum_{i=1}^{\infty}e^{-\delta T_i}\left|\log\frac{X_{T_i}}{X_{T_{i-1}}^{\gamma}}\right|
=&
     \sum_{i=1}^{\infty}e^{-\delta T_i}\left(\log\frac{X_{T_i}}{X_{T_{i-1}}^{\gamma}}\right)_{+}
    +
    \sum_{i=1}^{\infty}e^{-\delta T_i}\left(\log\frac{X_{T_i}}{X_{T_{i-1}}^{\gamma}}\right)_{-}<\infty\,\,  \text{ almost surely}.\nonumber
\end{align}
Because an absolutely convergent series is conditionally convergent, one has
\begin{align}
\label{convergence.a.s.}
    \sum_{i=1}^{n}e^{-\delta T_i}\log\bigg(\frac{X_{T_i}}{X_{T_{i-1}}^{\gamma}}\bigg)\longrightarrow
    \sum_{i=1}^{\infty}e^{-\delta T_i}\log\bigg(\frac{X_{T_i}}{X_{T_{i-1}}^{\gamma}}\bigg) \text{ almost surely as } n\rightarrow\infty.
\end{align}
With the help of \eqref{Control function}, \eqref{control func.integrable} and \eqref{convergence.a.s.}, one can apply Lemma 5.3 in \cite{WYY23} and the dominated convergence theorem to the inequality \eqref{A+C+=D} to obtain that
\begin{eqnarray}
\mathbb{E}\left[\sum_{i=1}^{\infty}e^{-\delta T_i}\bigg(\log\frac{X_{T_i}}{X_{T_{i-1}}^{\gamma}}+h(Y_{T_{i}})\bigg)\right]
\hspace{-0.3cm}&=&\hspace{-0.3cm}
\mathbb{E}\left[\sum_{i=1}^{\infty}e^{-\delta T_i}\left(\log\frac{X_{T_i}}{X_{T_{i-1}}^{\gamma}}+h(Y_{T_i})\right)\right]
\nonumber\\
\hspace{-0.3cm}&\leq&\hspace{-0.3cm}
A^*(y)+C^*\log x.\nonumber
\end{eqnarray}
It then holds that
\begin{eqnarray}
V(x,y)=\sup_{X\in\mathcal{U}_0(x,y)}\mathbb{E}\left[\sum_{i=1}^{\infty}e^{-\delta T_i}\bigg(\log\frac{X_{T_i}}{X_{T_{i-1}}^{\gamma}}+h(Y_{T_{i}})\bigg)\right]\leq A^*(y)+C^*\log x.\nonumber
\end{eqnarray}
To show the reverse inequality, it is sufficient to show the existence of some admissible portfolio process $X^*$ such that $$\mathbb{E}\left[\sum_{i=1}^{\infty}e^{-\delta T_i}\left(\log\frac{X^*_{T_i}}{(X^*_{T_{i-1}})^{\gamma}}+h(Y_{T_i})\right)\right]=A^*(y)+C^*\log x.$$
By Proposition \ref{prop2.2}, given the choice of $X^*_{T_{n+1}}=X^*_{T_n}\frac{Z^{0}_{T_n}/B_{T_n}}{Z^{0}_{T_{n+1}}/B_{T_{n+1}}}$ for $n\geq0$, it holds that
\begin{eqnarray}
\hspace{-0.3cm}&&\hspace{-0.3cm}
\mathbb{E}\left[\left.\frac{1-\gamma e^{-\delta\tau}}{e^{\delta\tau}-1}\log\frac{X^*_{T_{n+1}}}{X^*_{T_n}}+e^{-\delta\tau}(A^*(Y_{T_{n+1}})+h(Y_{T_{n+1}}))\right|\mathcal{F}_{T_n}\right]
\nonumber\\
\hspace{-0.3cm}&=&\hspace{-0.3cm}
\sup_{X\in\Tilde{\mathcal{U}}_{T_n,T_{n+1}}(1,Y_{T_n})}\mathbb{E}\left[\frac{1-\gamma e^{-\delta\tau}}{e^{\delta\tau}-1}\log X_{T_{n+1}}+e^{-\delta\tau}(A^*(Y_{T_{n+1}})+h(Y_{T_{n+1}}))\Big|\mathcal{F}_{T_n}\right]=\Psi(A^*;Y_{T_n}).\nonumber
\end{eqnarray}
Define
$$D^*_n:=\sum_{i=1}^ne^{-\delta T_i}\bigg(\log\frac{X^*_{T_i}}{{(X^*_{T_{i-1}})}^{\gamma}}+h(Y_{T_{i}})\bigg)+e^{-\delta T_n}(A^*(Y_{T_n})+C^*\log X^*_{T_n}), \quad n\geq 0.$$
Using the same arguments leading to \eqref{3.6.log}, it is easy to conclude that
\begin{eqnarray}
\mathbb{E}[D^*_{n+1}|\mathcal{F}_{T_n}]=D^*_n+e^{-\delta T_n}\left[e^{-\delta\tau}\Psi(A^*;Y_{T_n})-A^*(Y_{T_n})\right]=D^*_n, \quad n\geq 0,\nonumber
\end{eqnarray}
which implies that $(D^*_n)_{n\geq0}$ is a $\{\mathcal{F}_{T_n}\}$-martingale. Hence, we have
\begin{eqnarray}
\mathbb{E}\left[\sum_{i=1}^{n}e^{-\delta T_i}\bigg(\log\frac{X^*_{T_i}}{(X^*_{T_{i-1}})^{\gamma}}+h(Y_{T_{i}})\bigg)\right]
\hspace{-0.3cm}&=&\hspace{-0.3cm}
A^*(y)+C^*\log x-e^{-\delta T_n}\mathbb{E}\left[A^*(Y_{T_n})+C^*\log X^*_{T_n}\right].\nonumber
\end{eqnarray}
By the dominated convergence theorem, Proposition \ref{fixed.log} and Lemma 5.3 in \cite{WYY23}, we have
\begin{eqnarray}
\mathbb{E}\left[\sum_{i=1}^{\infty}e^{-\delta T_i}\bigg(\log\frac{X^*_{T_i}}{(X^*_{T_{i-1}})^{\gamma}}+h(Y_{T_{i}})\bigg)\right]=A^*(y)+C^*\log x.\nonumber
\end{eqnarray}
Therefore, the first claim holds. To prove the second claim, we only need to verify that $X^*_t=x\frac{B_t}{Z^0_t}$ satisfies $dX^*_t=[r+(\mu(Y_t)-r)\pi^*_t]X^*_tdt+\sigma(Y_t)\pi^*_tX^*_tdW_{1t}$. To this end, we observe that 
$$d\frac{B_t}{Z^0_t}=\frac{B_t}{Z_t^0}\left[rdt+\sigma(Y_t)\pi^*_tdW_{1t}+(\mu(Y_t)-r)\pi^*_tdt\right],\quad t\in[0,\infty).$$
Hence, we arrive at 
$$d\left[x\frac{B_t}{Z^0_t}\right]=[r+(\mu(Y_t)-r)\pi^*_t]X^*_tdt+\sigma(Y_t)\pi^*_tX^*_tdW_{1t},$$
which yields that the second claim holds.
\end{proof}

\ \\
\textbf{Acknowledgements}:  W. Wang is supported by the National Natural Science Foundation of China under no. 12171405 and no. 11661074. 
X. Yu is supported by the Hong Kong RGC General Research Fund (GRF) under grant no. 15304122 and grant no. 15306523.

\end{document}